\documentclass[11pt]{article}

\usepackage{graphicx}
\usepackage{amsfonts}
\usepackage{amssymb}
\usepackage{amsmath}
\usepackage{verbatim}
\usepackage{float}
\usepackage{scrextend}
\usepackage{xspace}
\usepackage{setspace}
\usepackage{algpseudocode}
\usepackage{tipa}
\usepackage{amsthm}
\usepackage{fullpage}
\usepackage[british]{babel}
\usepackage{color}

\newtheorem{definition}{Definition}[section]
\newtheorem{theorem}{Theorem}[section]

\newtheorem{lemma}{Lemma}[section]

\newtheorem{observation}{Observation}[section]
\newtheorem{corollary}{Corollary}[section]

\usepackage{fancyhdr}

\fancypagestyle{fltplain}{%
  \fancyhf{}
  \fancyfoot[C]{\iffloatpage{}{\thepage}}
}
\floatstyle{ruled}
\newfloat{algorithm}{tp}{lop}
\floatname{algorithm}{Algorithm}

\newcommand{\ie}{{\em i.e.,}\xspace}
\newcommand{\aka}{{\em a.k.a.}\xspace}
\newcommand{\eg}{{\em e.g.,}\xspace}
\newcommand{\etc}{{\em etc.}}

\newcommand{\PEF}{$\mathbb{SELF}$-$\mathbb{STAB\_PEF}\_3$}
\newcommand{\PEFR}{$\mathbb{SELF}$-$\mathbb{STAB\_PEF}\_2$}

\title{Self-Stabilizing Robots in\\Highly Dynamic Environments\thanks{A preliminary 
version of this work appears in the proceedings of the 18th International Symposium 
on Stabilization, Safety, and Security of Distributed Systems (SSS 2016) \cite{BDD16}.\newline 
\indent This work has been partially supported by the ANR project ESTATE and was initiated while 
the second author was visiting professor at UPMC Sorbonne Universit\'es.\newline
\indent The authors want to thank Florence Lev\'e from Universit\'e de Picardie-Jules Verne (France) for her help on the proofs of Lemmas \ref{separation_possible} and \ref{separation_possible_bis}.}}

\renewcommand*{\thefootnote}{\fnsymbol{footnote}}

\author{
Marjorie Bournat\footnotemark[2]
\and
Ajoy K. Datta\footnotemark[3]
\and
Swan Dubois\footnotemark[2]
}

\date{}

\begin{document}

\maketitle


   \begin{abstract}
This paper deals with the classical problem of exploring a
ring by a cohort of synchronous robots. We focus on the
perpetual version of this problem in which it is required
that each node of the ring is visited by a robot infinitely 
often. 

The challenge in this paper is twofold. First, we assume that 
the robots evolve in a highly dynamic ring, \ie edges may appear 
and disappear unpredictably without any recurrence, periodicity, 
nor stability assumption. The only assumption we made
(known as temporal connectivity assumption)
is that each node is infinitely often reachable from any
other node. 
Second, we aim at providing a self-stabilizing
algorithm to the robots, \ie the algorithm must guarantee
an eventual correct behavior regardless of the initial state and
positions of the robots.

In this harsh environment, our contribution is to fully characterize,
for each size of the ring, the necessary and sufficient number of robots 
to solve deterministically the problem.
\end{abstract}

\footnotetext[2]{UPMC Sorbonne Universit\'es, CNRS, Inria, LIP6 UMR 7606, France}

\footnotetext[3]{University of Nevada, Las Vegas, United States}

\renewcommand*{\thefootnote}{\arabic{footnote}}
\setcounter{footnote}{0}

\section{Introduction}\label{sec:intro}

   We consider a cohort of autonomous and synchronous robots that are
equipped with motion actuators and sensors, but that are otherwise
unable to communicate \cite{SY99}. They evolve in a \emph{discrete
environment},
represented by a graph, where the nodes represent the
possible locations of robots and the edges the possibility
for a robot to move from one location to another. Refer to \cite{PRT11}
for a survey of results in this model. One fundamental problem is the
\emph{exploration} of graphs by robots. Basically, each node of the
graph has to be visited by at least one robot. There exist several
variants of this problem depending on whether the robots are required
to stop once they completed the exploration of the graph or not.

Typically, the environment of the robots is modeled by a \emph{static}
undirected connected graph meaning that both vertex and edge sets do not evolve with time.
In this paper, we consider {\em dynamic} environments that may change over time, 
for instance, a transportation network, a building in which doors are closed and 
open over time, or streets that are closed over time due to work in process or 
traffic jam in a town. More precisely, we consider dynamic graphs \cite{XFJ03,CFQS12}
in which edges may appear and disappear unpredictably without any stability, recurrence, 
nor periodicity assumption. However, to ensure that the problem is not trivially 
unsolvable, we made the assumption that each node is infinitely often 
reachable from any other one through a {\em temporal path} (\aka~{\em journey}~\cite{CFQS12}). 
In the following, such dynamic graphs are indifferently called \emph{highly dynamic} 
or \emph{connected-over-time}.


As in other distributed systems, \emph{fault-tolerance} is a central issue
in robot networks. Indeed, it is desirable that the misbehavior of some
robots does not prevent the whole system to reach its objective.
\emph{Self-stabilization} \cite{D74,D00,T09} is a versatile technique to
tolerate
\emph{transient} (\ie of finite duration) faults. After the occurrence of a
catastrophic failure that may take the system to some arbitrary
global state, self-stabilization guarantees recovery to a correct behavior
in finite time without external (\ie human) intervention. In the
context of robot networks, that implies that the algorithm must guarantee
an eventual correct behavior regardless of the initial state and positions of
the robots.

Our objective in this paper is twofold. First, we want to investigate for the 
first time the problem of exploration of a highly dynamic graph by a cohort of 
self-stabilizing deterministic robots. Second, we aim at characterizing, for 
the specific case of the ring, the necessary and sufficient number of robots
to perform this task (in function of the size of the ring).

\paragraph{Related Work.} Since the seminal work of Shannon \cite{S51},
exploration of graphs by a cohort of robots has been extensively studied. There
exist mainly three variants of the problem: $(i)$ \emph{exploration with
stop}, where robots are required to detect the end of the exploration, then
stop moving (\eg \cite{FIPS07}); $(ii)$ \emph{exploration with return},
where robots must come back to their initial location once the exploration
completed (\eg \cite{DFKP04}); and $(iii)$ \emph{perpetual exploration}, where
each node has to be infinitely often visited by some robots (\eg
\cite{BBMR08}). Even if we restrict ourselves to deterministic approaches,
there exist numerous solutions to these problems depending on the topology
of the graphs to explore (\eg ring-shaped \cite{FIPS07}, line-shaped \cite{FIPS11},
tree-shaped \cite{FIPS10}, or arbitrary network \cite{CFMS10}), and
the assumptions made on robots (\eg limited range of visibility \cite{DLLP13},
common sense of orientation \cite{BMPT10}, \etc).

Note that all the above work consider only static graphs.
Recently, some work dealt with the exploration of dynamic graphs. 
We review them in the following.

The first two papers \cite{FMS09,IW11} focused on the exploration (with stop) by a 
single agent of so-called \emph{PV-graphs}. A PV-graph is a very specific 
kind of dynamic graph in which a set of entities (called carriers) infinitely often 
move in a predetermined way, each of them periodically visiting a subset of nodes 
of the graph. An agent (controlled by the algorithm) is initially located at a node 
and can move from node to node only using a carrier. This is a relevant model for 
transportation networks. In this context, these two papers study the necessity and 
the sufficiency of various assumptions (like the anonymity of nodes, the 
knowledge of the size of the network, 
or of a bound on the periodicity of the carriers, \etc) as well as their impact on the complexity 
of the exploration. The main difference between these two works lies on the assumption 
whether the agent is able to wait a carrier on nodes \cite{IW11} or not \cite{FMS09}.

A second line of research \cite{IW13,IKW14,DDFS16} considers another
restriction on dynamicity by targeting \emph{$T$-interval-connected graphs}, \ie
the graph is connected at each step and there exists a stability of this
connectivity in any interval of time of length $T$ \cite{KLO10}.
The two first papers \cite{IW13,IKW14} investigate an \emph{off-line} version of the
exploration meaning that the single agent knows in advance the evolution of the graph
over time and uses it to compute its route before the beginning of the execution.
In both papers, the authors provide lower and upper bounds on the exploration time 
in this context. The first one focuses on ring-shaped graphs, the second one on
cactus-shaped (\ie trees of rings). Finally, \cite{DDFS16} deals with the exploration
with stop of $1$-interval-connected rings by several robots. The authors study the impact
of numerous assumptions (like the synchrony assumption, the anonymity of the graph, the
chirality of robots, or the knowledge of some characteristic of the graph) on the 
solvability of the problem depending on the number of robots involved.
They particularly show that these assumptions may influence
the capacity of robots to \emph{detect} the end of the exploration and hence to 
systematically terminates their execution or not. 

In summary, previous work on exploration of dynamic graphs restricts strongly the 
dynamic of the considered graph. The notable exception is a recent work on perpetual 
exploration of highly dynamic rings \cite{BDP17}. This paper shows that, three 
(resp. two) synchronous anonymous robots are necessary and sufficient to 
perpetually explore highly dynamic rings of size greater (resp. equals) to four (resp. three).
Nonetheless, algorithms from \cite{BDP17} do not tolerate initial memory corruption nor 
arbitrary initial positions of robots. In other words, they are not self-stabilizing.
Moreover, to the best of our knowledge, there exist no self-stabilizing algorithm for 
exploration either in a static or a dynamic environment. Note that there exist
such fault-tolerant solutions in static graphs to other problems (\eg naming and 
leader election \cite{BPT07}).

\paragraph{Our Contribution.} The main contribution of this paper is to prove that
the necessary and sufficient numbers of robots for perpetual exploration of highly 
dynamic rings exhibited in \cite{BDP17} also hold in a self-stabilizing setting 
\emph{at the price of the loss of anonymity of robots}.

More precisely, this result is achieved through the following technical achievements.
Section \ref{sec:necessary} presents two impossibility results establishing that
at least two (resp. three) self-stabilizing robots are necessary to perpetually explore
highly dynamic rings of size greater than 3 (resp. 4) \emph{even if robots are not anonymous}.
Note that these necessity results are not implied by the ones of \cite{BDP17} (that focuses
on anonymous robots).  Then, Sections \ref{sec:algo3robots}
and \ref{sec:algo2robots} present and prove two algorithms showing the sufficiency 
of these conditions.

\section{Model}\label{sec:model}

In this section, we propose an extension of the classical model of robot networks 
in static graphs introduced by \cite{KMP06} to the context of dynamic graphs.
 
\paragraph{Dynamic graphs.} In this paper, we consider the model of \emph{evolving
graphs} introduced in \cite{XFJ03}. We hence consider the time as discretized 
and mapped to $\mathbb{N}$. An evolving graph $\mathcal{G}$ is an ordered 
sequence $\{G_{0},G_{1}, G_{2}, \ldots\}$ of subgraphs of a given static graph 
$G=(V,E)$. The (static) graph $G$ is called the \emph{footprint} of $\mathcal{G}$.
In the following, we restrict ourselves to evolving graphs whose footprints are 
anonymous, bidirectional, unoriented, and simple graphs.
For any $i\geq 0$, we have $G_{i} = (V, E_{i})$ and we say that
the edges of $E_{i}$ are \emph{present} in $\mathcal{G}$ at time $i$. The 
\emph{underlying graph} of $\mathcal{G}$, denoted $U_\mathcal{G}$, is the static graph
gathering all edges that are present at least once in $\mathcal{G}$ (\ie 
$U_\mathcal{G}=(V,E_\mathcal{G})$ with $E_\mathcal{G}=\bigcup_{i=0}^{\infty}E_i)$. 
An \emph{eventual missing edge} is an edge of $E_\mathcal{G}$ such that there 
exists a time after which this edge is never present in $\mathcal{G}$. A 
\emph{recurrent edge} is an edge of $E_\mathcal{G}$ that is not eventually missing.
The \emph{eventual underlying graph} of $\mathcal{G}$, denoted $U_\mathcal{G}^\omega$, 
is the static graph gathering all recurrent edges of $\mathcal{G}$ 
(\ie $U_\mathcal{G}^\omega=(V,E_\mathcal{G}^\omega)$ where 
$E_\mathcal{G}^\omega$ is the set 
of recurrent edges of $\mathcal{G}$). In this paper, we chose to make minimal
assumptions on the dynamicity of our graph since we restrict ourselves on 
\emph{connected-over-time} evolving graphs. The only constraint we impose on evolving 
graphs of this class is that their eventual underlying graph is connected \cite{DKP15} 
(intuitively, that means that any node is infinitely often reachable from 
any other one). For the sake of the proof, we also consider the weaker class 
of \emph{edge-recurrent} evolving graphs where the eventual underlying graph is connected
and matches to the footprint. In the following, we consider
only connected-over-time evolving graphs whose footprint is a ring of arbitrary size
called connected-over-time rings for simplicity. We call $n$ the size of the ring. Although the ring is unoriented,
to simplify the presentation and discussion, we, as external observers, 
distinguish between the clockwise and the
counter-clockwise direction in the ring.

For the sake of some proofs in this paper, we need to introduce an operator denoted $\backslash$ that
removes some edges of an evolving graph for some time ranges. More formally,
from an evolving graph $\mathcal{G}=\{(V, E_0), (V, E_1),(V, E_2),$ $\ldots\}$, 
we define the evolving graph 
$\mathcal{G} \backslash \{(e_{1}, \tau_{1}), \ldots (e_{k}, \tau_{k})\}$ (with 
for any $i \in \{1, \ldots, k\}$, $e_{i} \in E$ and $\tau_{i} \subseteq \mathbb{N}$) as the 
evolving graph $\{(V, E_0'), (V, E_1'), (V, E_2'),\ldots\}$ such that:
$\forall t\in\mathbb{N},\forall e\in E_\mathcal{G},
e\in E_t' \Leftrightarrow e\in E_t\wedge(\forall i\in\{1,\ldots,k\}, e\neq e_i \vee t\notin \tau_i)$.


\paragraph{Robots.} We consider systems of autonomous mobile entities called
robots moving in a discrete and dynamic environment modeled by an evolving graph
$\mathcal{G}=\{(V,E_1),(V,E_2)\ldots\}$, $V$ being a set of nodes representing the set of
locations where robots may be, $E_i$ being the set of bidirectional edges representing connections through
which robots may move from a location to another one at time $i$. Robots are uniform
(they execute the same algorithm), identified (each of them has a distinct identifier),
have a persistent memory but are unable to directly communicate with one
another by any means.
Robots are endowed with local strong multiplicity detection (\ie they are able to
detect the exact number of robots located on their current node).
They have no a priori knowledge about the ring they explore (size, diameter, dynamicity, \ldots)
nor on the robots (number, bound on size of identifiers\ldots). 
Finally, each robot has its own stable chirality (\ie each robot is able to locally label 
the two ports of its current node with \emph{left} and \emph{right} consistently 
over the ring and time but two different robots may not agree on this labeling). 
We assume that each robot has a variable $dir$ that stores a direction (either
\emph{left} or \emph{right}). At any time, we say that a robot points to
\emph{left} (resp. \emph{right}) if its variable $dir$ is equal
to this (local) direction. We say that 
a robot considers the clockwise (resp., counter-clockwise) direction if the (local) 
direction pointed to by this robot corresponds to the (global) direction seen by an 
external observer. 

\paragraph{Execution.} The configuration of the system at time $t$ (denoted $\gamma_t$)
captures the position (\ie the node where the robot is currently located) and the state 
(\ie the value of every variable of the robot) of each robot at a given time. 
We say that robots form a \emph{tower} on a node $v$ in $\gamma_t$ if at least two robots 
are co-located on $v$ in $\gamma_t$. Given an evolving graph 
$\mathcal{G}=\{G_{0},G_{1}, G_{2}, \ldots\}$, an algorithm $\mathcal{A}$, and an initial
configuration $\gamma_0$, the execution $\mathcal{E}$ of $\mathcal{A}$ on $\mathcal{G}$
starting from $\gamma_0$ is the infinite sequence $(G_0,\gamma_0),(G_1,\gamma_1),
(G_2,\gamma_2),\ldots$ where, for any $i\geq 0$, the configuration $\gamma_{i+1}$ is
the result of the execution of a synchronous round by all robots 
from $(G_i,\gamma_i)$ as explained below.

The round that transitions the system from $(G_i,\gamma_i)$ to $(G_{i+1},\gamma_{i+1})$ 
is composed of three atomic and synchronous phases: Look, Compute, Move. During 
the Look phase, each robot gathers information about its environment in $G_i$. More
precisely, each robot updates the value of the following local predicates:
$(i)$ $Number\-Of\-Ro\-bots\-On\-No\-de()$ that returns the exact number of robots present
at the node of the robot;
$(ii)$ $Exists\-Edge\-On\-Left()$ that returns true if an edge in the left 
direction of the robot is present, false otherwise;  
$(iii)$ $Exists\-Edge\-On\-Right()$ that returns true if an edge in the right 
direction of the robot is present, false otherwise;
$(iv)$ $Exists\-Adjacent\-Edge()$ returns true if an edge adjacent to
the current node of the robot is present, false otherwise.
During the Compute phase, each robot executes the algorithm $\mathcal{A}$
that may modify some of its variables (in particular $dir$) depending on its 
current state and on the values of the predicates updated during the Look phase. 
Finally, the Move phase consists of moving each robot trough one edge 
in the direction it points to if there exists an edge in that direction, otherwise,
\ie if the edge is missing at that time, the robot remains at its current node.
Note that the $i^{th}$ round is entirely executed on $G_{i}$ and that the 
transition from $G_{i}$ to $G_{i + 1}$ occurs only at the end of this round.
We say that a robot
is \emph{edge-activated} during a round if there exists at least one edge 
adjacent to its location during that round. To simplify the pseudo-code of the 
algorithms, we assume that the robots have access to two predicates: 
$Exists\-Edge\-On\-Current\-Dire\-ction()$ (that returns true if an edge is 
present at the direction currently pointed by the robot, false otherwise) and 
$Exists\-Edge\-On\-Oppo\-si\-te\-Dire\-ction()$ (that returns true if an edge is
present in the direction opposite to the one currently pointed by the robot, 
false otherwise). Both of these two predicates 
depend on the values of the predicates $ExistsRightEdge()$ and 
$ExistsLeftEdge()$, and on the value of the variable $dir$.

\paragraph{Self-Stabilization.} Intuitively, a self-stabilizing 
algorithm is able to recover in a finite time a correct
behavior from any arbitrary initial configuration (that captures the effect of an
arbitrary transient fault in the system). More formally, an algorithm $\mathcal{A}$ 
is \emph{self-stabilizing} for a problem on a class of evolving graphs $\mathcal{C}$ 
if and only if it ensures that, for any configuration $\gamma_0$, the execution of 
$\mathcal{A}$ on any $\mathcal{G}\in\mathcal{C}$ starting from $\gamma_0$ contains 
a configuration $\gamma_i$ such that the execution of $\mathcal{A}$ on $\mathcal{G}$
starting from $\gamma_i$ satisfies the specification of the problem. Note that,
in the context of robot networks, this definition implies that robots must tolerate
both arbitrary initialization of their variables and arbitrary initial positions
(in particular, robots may be stacked in the initial configuration).

\paragraph{Perpetual Exploration.} Given an evolving graph $\mathcal{G}$, 
a perpetual exploration algorithm guarantees that every node of
$\mathcal{G}$ is infinitely often visited by at least one robot (\ie a robot 
is infinitely often located at every node of $\mathcal{G}$). Note that this 
specification does not require that every robot visits infinitely often every 
node of $\mathcal{G}$.

\section{Necessary Number of Robots}\label{sec:necessary}

   This section is devoted to the proof of the necessity of two (resp. three) 
self-stabilizing identified robots to perform perpetual exploration of highly 
dynamic rings of size at least 3 (resp. 4). To reach this goal, we provide two
impossibility results. 

First, we prove (see Theorem~\ref{no_perpetual_exploration_two_robots})
that two robots with distinct identifiers are not able to perpetually explore in a 
self-stabilizing way connected-over-time rings of size greater than 4.
Then, we show that we can borrow arguments from \cite{BDP17} to prove Theorem~
\ref{no_perpetual_exploration_one_robot} that states that only one robot cannot
complete the self-stabilizing perpetual exporation of connected-over-time rings 
of size greater than 3.


\subsection{Highly Dynamic Rings of Size $4$ or More}

The proof of Theorem~\ref{no_perpetual_exploration_two_robots} makes use of a generic 
framework proposed in \cite{BDKP16}. Note that, even if this generic framework 
is designed for another model (namely, the classical message passing model), it 
is straightforward to borrow it for our current model. Indeed, its proof only 
relies on the determinism of algorithms and indistinguishability of dynamic
graphs, these arguments being directly translatable in our model. We present
briefly this framework here. The interested reader is referred to \cite{BDKP16}
for more details. 

This framework is based on a theorem that ensures that, if we take a sequence of 
evolving graphs with ever-growing common prefixes (that hence converges to the
evolving graph that shares all these common prefixes), then the sequence of 
corresponding executions of any deterministic algorithm also converges. 
Moreover, we are able to describe the execution to which it converges as the 
execution of this algorithm on the evolving graph to which the sequence
converges. This result is useful since it allows us to construct counter-example
in the context of impossibility results. Indeed, it is sufficient to construct
an evolving graphs sequence (with ever-growing common prefixes) and to prove 
that their corresponding execution violates the specification of the problem 
for ever-growing time to exhibit an execution that never satisfies the
specification of the problem.

In order to build the evolving graphs sequence suitable for the proof of our
impossibility result, we need the following technical lemma.

\begin{lemma} \label{lemma_modification_direction_bis}
 Let $\mathcal{A}$ be a self-stabilizing deterministic perpetual exploration
 algorithm in connected-over-time rings of size $4$ or more using 2 robots 
 $r_{1}$ and $r_{2}$ with distinct identifiers. Any execution of $\mathcal{A}$
 satisfies: For any time $t$, for any states $s_{1}$ and $s_{2}$, for any
 distinct identifiers $id_{1}$ and $id_{2}$, it exists $t'$ such that if
 $r_{1}$, of identifier $id_{1}$, is on node $u_{1}$ in state $s_{1}$, and
 $r_{2}$, of identifier $id_{2}$, is on node $u_{2}$ in state $s_{2}$ such that 
 there exists only one adjacent edge to each position of the robots continuously
 present from time $t$ to time $t'$, then $r_{1}$ and/or $r_{2}$ moves at time
 $t'$. This lemma holds even if the robots have the same chirality.
\end{lemma}

\begin{proof}
 Consider a self-stabilizing algorithm $\mathcal{A}$ that solves 
 deterministically the perpetual exploration problem for connected-over-time 
 rings of size $4$ or more using two robots with distinct identifiers.
 Let $\mathcal{G} = \{G_{0}, G_{1}, \ldots\}$ be a connected-over-time ring
 whose footprint $G$ is a ring of size $4$ or more, and such that
 $\forall i \in \mathbb{N}, G_{i} = G$. Let $r_{1}$ and $r_{2}$ be two robots 
 executing $\mathcal{A}$ on $\mathcal{G}$. 
 
 By contradiction, assume that there exists a time $t$, two states $s_{1}$ and
 $s_{2}$, and two distinct identifiers $id_{1}$ and $id_{2}$ such that at time
 $t$, $r_{1}$ of identifier $id_{1}$ is in state $s_{1}$ on a node $u_{1}$ of
 $\mathcal{G}$, $r_{2}$ of identifier $id_{2}$ is in state $s_{2}$ on a node 
 $u_{2}$ (it is possible to have $u_{2} = u_{1}$) of $\mathcal{G}$, such that,
 for any time $t' \geq t$, if there exists only one adjacent edge to each 
 position of the robots continuously present from time $t$ to time $t'$, then
 none of the robots move. 
 
 Consider the graph $\mathcal{G}'$ such that 
 $G' = G$ and such that 
 $\mathcal{G}' = \mathcal{G} \backslash \{(e, \{0, \ldots, +\infty\})\}$, where $e$
 is the edge linking $u_{1}$ and its adjacent node in the clockwise direction.
 Note that $\mathcal{G}'$ is a connected-over-time ring, since it only possesses
 one eventual missing edge.
  
 In the case where at time $t$ in $\mathcal{G}$, $r_{1}$ and $r_{2}$ are on the
 same node, since $\mathcal{A}$ is a self-stabilizing algorithm, we can 
 initially place $r_{1}$ and $r_{2}$ on node $u_{1}$ of $\mathcal{G}'$ in state
 $s_{1}$ and $s_{2}$ respectively.
 
 In the case where at time $t$ in $\mathcal{G}$, $r_{1}$ and $r_{2}$ are not on
 the same node, since $\mathcal{A}$ is a self-stabilizing algorithm, we can 
 initially place $r_{1}$ on node $u_{1}$ of $\mathcal{G}'$ in state $s_{1}$ and
 $r_{2}$ on the adjacent node of node $u_{1}$ in the clockwise direction in
 state $s_{2}$.

 In these two cases, by construction, there is only one adjacent edge to each 
 position of the robots continuously present from time $0$ to $+\infty$, and 
 $r_{1}$ and $r_{2}$ are respectively in state $s_{1}$ and $s_{2}$. Then, by 
 assumption, $r_{1}$ and $r_{2}$ does not leave their respective nodes after
 time $0$. As $\mathcal{G}'$ counts 4 nodes or more, we obtain a
 contradiction with the fact that $\mathcal{A}$ is a self-stabilizing algorithm
 solving deterministically the perpetual exploration problem for 
 connected-over-time rings of size $4$ or more using two robots.
\end{proof}

\begin{theorem} \label{no_perpetual_exploration_two_robots}
 There exists no deterministic algorithm satisfying the perpetual exploration
 specification in a self-stabilizing way on the class of connected-over-time 
 rings of size $4$ or more with two fully synchronous robots possessing distinct
 identifiers.
\end{theorem}

\begin{proof}
 By contradiction, assume that there exists a deterministic algorithm
 $\mathcal{A}$ satisfying the perpetual exploration specification in a
 self-stabilizing way on any connected-over-time rings of size $4$ or 
 more using two robots $r_{1}$ and $r_{2}$ possessing distinct identifiers.

 Consider the connected-over-time graph $\mathcal{G}= \{G_{0}, G_{1}, \ldots\}$
 whose footprint $G$ is a ring of size strictly greater than 3 and such that 
 $\forall i \in \mathbb{N}, G_{i} = G$.
 
 Consider four nodes $u$, $v$, $w$ and $x$ of $\mathcal{G}$, such that node $v$
 is the adjacent node of $u$ in the clockwise direction, $w$ is the adjacent
 node of $v$ in the clockwise direction, and $x$ is the adjacent node of $w$ in
 the clockwise direction. We denote respectively $e_{ur}$ and $e_{ul}$ the
 clockwise and counter clockwise adjacent edges of $u$, $e_{vr}$ and $e_{vl}$ 
 the clockwise and counter clockwise adjacent edges of $v$, $e_{wr}$ and 
 $e_{wl}$ the clockwise and counter clockwise adjacent edges of $w$, and 
 $e_{xr}$ and $e_{xl}$ the clockwise and counter clockwise adjacent edges of
 $x$. Note that $e_{ur} = e_{vl}$, $e_{vr} = e_{wl}$, and $e_{wr} = e_{xl}$.
 
 Let $\varepsilon$ be the execution of $\mathcal{A}$ on $\mathcal{G}$ starting
 from the configuration where $r_{1}$ (resp. $r_{2}$) is located on node $v$ 
 (resp. $w$).


 \begin{figure} \vspace{-2cm}

 		 \centerline{\includegraphics[scale=0.88]{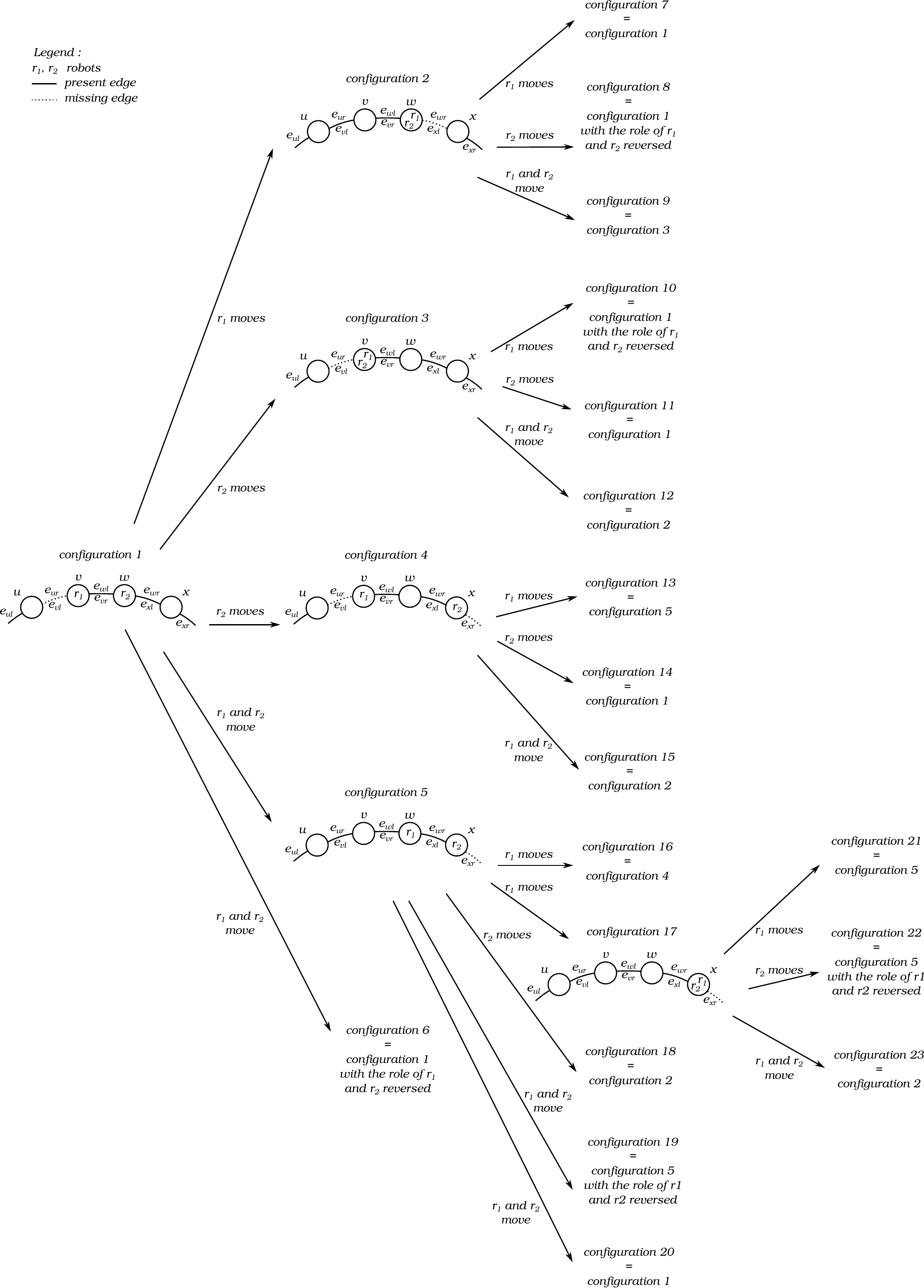}}

 	\caption{Construction of $\mathcal{G}'$ in proof of Lemma 
            \ref{lemma_modification_direction_bis}.} \label{impossibility_bis}
 \end{figure}

 
 We construct a sequence of connected-over-time graphs 
 ($\mathcal{G}_{m}$)$_{m \in \mathbb{N}}$ such that
 $\mathcal{G}_{0} = \mathcal{G}$ and for any $i \geq 0$, if $\mathcal{G}_{i}$ 
 exists, it is a connected-over-time ring, and at time $t_{i}$ only nodes among 
 $\{v, w, x\}$ have been visited, then we can define $\mathcal{G}_{i + 1}$ as
 shown on Figure~\ref{impossibility_bis}, and as explained in the following
 paragraph (denote by $\varepsilon_i$ the execution of $\mathcal{A}$ on
 $\mathcal{G}_{i}$ starting from the same configuration as $\varepsilon$).
 
 Since $\mathcal{G}_{i}$ is a connected-over-time ring, and since $\mathcal{A}$
 is a deterministic algorithm solving the perpetual exploration algorithm in a 
 self-stabilizing way on connected-over-time rings of size $4$ or more using 2
 robots possessing distinct identifiers, when the configuration $\gamma_i$ at 
 time $t_{i}$ is such that the two robots possess exactly one adjacent edge
 present, we use Lemma~\ref{lemma_modification_direction_bis} to exhibit time
 $t_{i}' \geq t_{i}$ such that if the configuration $\gamma_i$ last from time 
 $t_{i}$ to time $t_{i}'$, then one or both of the robots move. Similarly, since 
 $\mathcal{G}_{i}$ is a connected-over-time ring, and since $\mathcal{A}$
 is a deterministic algorithm solving the perpetual exploration algorithm in a 
 self-stabilizing way on connected-over-time rings of size $4$ or more using 2
 robots possessing distinct identifiers, when the configuration $\gamma_i$ at 
 time $t_{i}$ is such that there is only one missing edge, and that only one
 robot is adjacent to this missing edge, then we can also exhibit a time
 $t_{i}'$ at which at least one of the robots move. Indeed, if this
 configuration lasts from time $t_{i}$ to time $+\infty$, $\mathcal{G}_{i}$ is a 
 connected-over-time ring, and if none of the robot move in this configuration, 
 the exploration cannot be solved. Therefore such a time $t_{i}'$ exists. In 
 the following we show how we construct the dynamic graph $\mathcal{G}_{i + 1}$,
 in function of $t_{i}'$ and $\mathcal{G}_{i}$. If the two robots are on two 
 distinct nodes in $\mathcal{G}_{i}$ then:
 \begin{itemize}
  \item if one of the robot is on node $v$ and the other robot is on node $w$ 
  then we define $\mathcal{G}_{i + 1}$ such that $\mathcal{G}_{i + 1}$ and 
  $\mathcal{G}_{i}$ have the same footprint, and $\mathcal{G}_{i + 1} = 
  \mathcal{G}_{i} \backslash \{(e_{vl}, \{t_{i}, \ldots, t_{i}'\})\}$.
  \item if one of the robot is on node $x$ and the other robot is on node $w$
  then we define $\mathcal{G}_{i + 1}$ such that $\mathcal{G}_{i + 1}$ and 
  $\mathcal{G}_{i}$ have the same footprint, and $\mathcal{G}_{i + 1} = 
  \mathcal{G}_{i} \backslash \{(e_{xr}, \{t_{i}, \ldots, t_{i}'\})\}$.
  \item if one of the robot is on node $v$ and the other robot is on node $x$
  then we define $\mathcal{G}_{i + 1}$ such that $\mathcal{G}_{i + 1}$ and 
  $\mathcal{G}_{i}$ have the same footprint, and $\mathcal{G}_{i + 1} = 
  \mathcal{G}_{i} \backslash \{(e_{vl}, \{t_{i}, \ldots, t_{i}'\}), 
  (e_{xr}, \{t_{i}, \ldots, t_{i}'\})\}$.
 \end{itemize}

 If the two robots are on the same node in $\mathcal{G}_{i}$ then:
 \begin{itemize}
  \item if the two robots are on node $v$, then we define $\mathcal{G}_{i + 1}$ 
  such that $\mathcal{G}_{i + 1}$ and $\mathcal{G}_{i}$ have the same footprint,
  and $\mathcal{G}_{i + 1} = 
  \mathcal{G}_{i} \backslash \{(e_{vl}, \{t_{i}, \ldots, t_{i}'\})\}$.
  \item if the two robots are on node $w$ then we define $\mathcal{G}_{i + 1}$ 
  such that $\mathcal{G}_{i + 1}$ and $\mathcal{G}_{i}$ have the same footprint,
  and $\mathcal{G}_{i + 1} = 
  \mathcal{G}_{i} \backslash \{(e_{wr}, \{t_{i}, \ldots, t_{i}'\})\}$.
  \item if the two robots are on node $x$, then we define $\mathcal{G}_{i + 1}$ 
  such that $\mathcal{G}_{i + 1}$ and $\mathcal{G}_{i}$ have the same footprint,
  and $\mathcal{G}_{i + 1} =  
  \mathcal{G}_{i} \backslash \{(e_{xr}, \{t_{i}, \ldots, t_{i}'\})\}$.
 \end{itemize}

 Note that $\mathcal{G}_{i}$ and $\mathcal{G}_{i + 1}$ are indistinguishable for
 robots before time $t_{i}$. This implies that, at time $t_{i}$, $r_{1}$ and 
 $r_{2}$ are on the same node in $\varepsilon_{i}$ and in $\varepsilon_{i + 1}$.
 By construction of $t_{i}'$, either $r_{1}$ or $r_{2}$ or both of the two
 robots move at time $t_{i}'$ in $\varepsilon_{i + 1}$.  $\mathcal{G}_{i + 1}$ 
 is a connected-over-time ring (since it is indistinguishable from $\mathcal{G}$
 after $t_{i}' + 1$). Moreover, even if one or both of the robots move during 
 the Move phase of time $t_{i}'$, at time $t_{i}' + 1$ the robots are still on
 nodes among $\{v, w, x\}$, by assumption on $\mathcal{G}_{i}$ and since from 
 time $t_{i}$ to time $t_{i}'$ the edges permitting to go on a node other than
 the nodes among $\{v, w, x\}$ are missing. 
 
 Let $t_{i + 1} = t_{i}' + 1$. Then we can construct recursively each dynamic
 ring of ($\mathcal{G}_{m}$)$_{m \in \mathbb{N}}$ by applying the argues above 
 on all the possible configurations reached by the movements of the robots at 
 time $t_{i}'$ on $\mathcal{G}_{i + 1}$. 
 
 Note that the recurrence can be initiated, since $\mathcal{G}_{0}$
 exists, is a connected-over-time ring and that at time $t_{0} = 0$ only nodes 
 among $\{v, w, x\}$ have been visited. In other words, 
 ($\mathcal{G}_{m}$)$_{m \in \mathbb{N}}$ is well-defined.
 
 We can then define the evolving graph $\mathcal{G}_{\omega}$ such that
 $\mathcal{G}_{\omega}$ and $\mathcal{G}_{0}$ have the same footprint, and such
 that for all $i \in \mathbb{N}$, $\mathcal{G}_{\omega}$ shares a common prefix
 with $\mathcal{G}_{i}$ until time $t_{i}'$.
 
 Note that among the configurations presented in Figure~\ref{impossibility_bis},
 only Configuration 4 contains 2 missing edges. However, if this configuration
 is reached in $\mathcal{G}_{\omega}$ the following configuration reached in
 $\mathcal{G}_{\omega}$ is either Configuration 13 or Configuration 14 or 
 Configuration 15. Since these three configurations possess only one missing
 edge, this implies that $\mathcal{G}_{\omega}$ is a connected-over-time ring.
 
 As the sequence $(t_{m}$)$_{m \in \mathbb{N}}$ is increasing by construction,
 this implies that the sequence ($\mathcal{G}_{m}$)$_{m \in \mathbb{N}}$
 converges to $\mathcal{G}_{\omega}$.

 Applying the theorem of \cite{BDKP16}, we obtain that, until time $t_{i}'$, the 
 execution of $\mathcal{A}$ on $\mathcal{G}_{\omega}$ is identical to the one on 
 $\mathcal{G}_{i}$. This implies that, executing $\mathcal{A}$ on 
 $\mathcal{G}_{\omega}$ (whose footprint is a ring of size strictly greater than
 3), $r_{1}$ and $r_{2}$ only visit the nodes among $\{v, w, x\}$. This is 
 contradictory with the fact that $\mathcal{A}$ satisfies the perpetual
 exploration specification on connected over time rings of size strictly greater
 than 3 using 2 robots. 
\end{proof}

\subsection{Highly Dynamic Rings of Size $3$ or More}



In \cite{BDP17}, the authors prove (in Theorem V.1) that a single anonymous and
synchronous robot cannot perpetually explore connected-over-time rings of size 
$3$ or more in a fault-free setting. We can do two observations. First, any 
fault-free synchronous execution is possible in a self-stabilizing setting.
Second, in the case of a \emph{single} robot, the anonymous and the identified model
are equivalent.

These observations are sufficient to directly state the following result:

\begin{theorem} \label{no_perpetual_exploration_one_robot}
 There exists no deterministic algorithm satisfying the perpetual exploration
 specification in a self-stabilizing way on the class of connected-over-time 
 rings of size $3$ or more using one robot possessing an identifier.
\end{theorem}



\section{Sufficiency of Three Robots for $n\geq 4$}\label{sec:algo3robots}

   In this section, we present our self-stabilizing deterministic algorithm for 
the perpetual exploration of any connected-over-time ring of size greater than 4 with three robots. 
In this context, the difficulty to complete the exploration is twofold. First, 
in connected-over-time graphs, robots must deal with the possible existence
of some eventual missing edge (without the guarantee that such edge always exists).
Note that, in the case of a ring, there is at most one eventual missing edge
in any execution (otherwise, we have a contradiction with the connected-over-time
property). Second, robots have to handle the arbitrary initialization of the system
(corruption of variables and arbitrary position of robots).

\subsection{Presentation of the algorithm}

\paragraph{Principle of the algorithm.} The main idea behind our algorithm is
that a robot does not change its direction (arbitrarily initialized) while it is 
isolated. This allows robots to perpetually explore connected-over-time rings with no
eventual missing edge regardless of the initial direction of the robots. 

Obviously, this idea is no longer sufficient when there 
exists an eventual missing edge since, in this case, at least two robots will
eventually be stuck (\ie they point to an eventual missing edge that they are 
never able to cross) forever at one end of the eventual missing edge. When
two (or more) robots are located at the same node, we say that they form a tower.
In this case, our algorithm succeed (as we explain below) to ensure that at least 
one robot leaves the tower in a finite time. In this way, we obtain that, in a finite time,
a robot is stuck at each end of the eventual missing edge. These two robots 
located at two ends of the eventual missing edge play the role of ``sentinels'' 
while the third one (we call it a ``visitor'') visits other nodes of the ring in the following way. The 
``visitor'' keeps its direction until it meets one of these ``sentinels'',
they then switch their roles: After the meeting, the ``visitor'' still 
maintains the same direction (becoming thus a ``sentinel'') while the ``sentinel''
robot changes its direction (becoming thus a ``visitor'' until reaching 
the other ``sentinel'').

In fact, robots are never aware if they are actually stuck at an eventual missing 
edge or are just temporarily stuck on an edge that will reappear in a 
finite time. That is why it is important that the robots keep consider their directions 
and try to move forward while there is no meeting in order to track a possible 
eventual missing edge. Our algorithm only guarantees a convergence in a finite 
time towards a configuration where a robot plays the role of ``sentinel'' at each 
end of the eventual missing edge if such an edge exists. Note that, in 
the case where there is no eventual missing edge, this mechanism does not 
prevent the correct exploration of the ring since it is impossible for a robot 
to be stuck forever.

Our algorithm easily deals with the initial corruption of its variables. 
Indeed, all variables of a robot (at the exception of a counter and the variable
$dir$ whose initial respective values have no particular impact) store
information about the environment of this robot in the previous round it was 
edge-activated. These variables are updated each time a robot is edge-activated.
Since we consider connected-over-time rings, it can only exist one eventual
missing edge, therefore all robots are infinitely often edge-activated. The 
initial values of these variables are hence reset in a finite time.
The main difficulty to
achieve self-stabilization is to deal with the arbitrary initial position of robots.
In particular, the robots may initially form towers. In the worst case, all robots of a 
tower may be stuck at an eventual missing edge and be in the same state. They are then
unable to start the ``sentinels''/``visitor'' scheme explained above. Our algorithm 
needs to ``break'' such a tower in a finite time 
(\ie one robot must leave the node where the tower is located).
In other words, we tackle a classical problem of symmetry breaking. We succeed
by providing each robot with a function that returns, in a finite number of invocations,
different global directions to two robots of the tower based on the private 
identifier of the robot and without any communication among 
the robots. More precisely, this is done thanks to a transformation of the 
robot identifier: each bit of the binary representation
of the identifier is duplicated and we add the bits ``010'' at the end of the 
sequence of these duplicated bits. Then, at each invocation of the function,
a robot reads the next bit of this transformed identifier. If the robot reads
zero, it try to move to its left. Otherwise, it try to move to its right.
Doing so, in a finite number of invocation of this function, at least one robot 
leaves the tower.  
If necessary, we repeat this ``tower breaking'' scheme until we are able to start the 
``sentinels''/``visitor'' scheme.

The main difficulty in designing this algorithm is to ensure 
that these two mechanisms (``sentinels''/``visitor'' and ``tower breaking'') 
do not interfere with each other and prevent the correct exploration. We solve this problem
by adding some waiting at good time, especially before starting the procedure of tower breaking
by identifier to ensure that robots do not 
prematurely turn back and ``forget'' to explore some parts of the ring.

\paragraph{Formal presentation of the algorithm.} Before presenting formally our 
algorithm, we need to introduce the set of constants (\ie variables assumed to be not
corruptible) and the set of variables of each robot. We also introduce
three auxiliary functions.

As stated in the model, each robot has an unique identifier. We denote it by 
$id$ and represent
it in binary as $b_{1}b_{2}\ldots b_{|id|}$. We define, for the purpose 
of the ``breaking tower'' scheme, the constant $TransformedIdentifier$ by its binary
representation $b_{1}b_{1}b_{2}b_{2} \ldots b_{|id|}b_{|id|} 010$ 
(each bit of $id$ is duplicated and we add the three bits $010$ at the end). 
We store the
length of the binary representation of $TransformedIdentifier$ in the constant $\ell$ 
and we denote its $i$th bit by $TransformedIdentifier[i]$ for any $1\leq i\leq \ell$.

In addition to the variable $dir$ defined in the model, each robot has the
following three variables: 
$(i)$ the variable $i\in\mathbb{N}$ corresponds to an index to store 
the position of the last bit read from $TransformedIdentifier$;
$(ii)$ the variable $Number\-Ro\-bots\-Pre\-vious\-Edge\-Acti\-va\-tion\in\mathbb{N}$ stores the 
number of robots that were present at the node of the robot during the Look phase
of the last round where it was edge-activated; and
$(iii)$ the variable $Has\-Moved\-Pre\-vious\-Edge\-Acti\-va\-tion\in\{true,false\}$ indicates 
if the robot has crossed an edge during its last edge-activation. 

Our algorithm makes use of a function \textsc{Update} that updates the values of the 
two last variables according to the current environment of the robot each time it is 
edge-activated. We provide the pseudo-code of this function in Algorithm 
\ref{algo:UpdateFunction}. Note that this function also allows us to deal with 
the initial corruption of the two last variables since it resets them in the 
first round where the robot is edge-activated.

We already stated that, whenever robots are stuck forming a tower, they make use of a 
function to ``break'' the tower in a finite time. The pseudo-code of this function 
\textsc{GiveDirection} appears in Algorithm \ref{algo:GiveDirection}. It
assigns the value $left$ or $right$ to the variable $dir$ of the robot depending on the 
$i$th bit of the value of $TransformedIdentifier$. The variable $i$ is incremented 
modulo $\ell$ (that implicitly resets this variable when it is corrupted)
to ensure that successive calls to \textsc{GiveDirection} will consider each bit of
$TransformedIdentifier$ in a round-robin way. As shown in the next section, 
this function guarantees that, if two robots 
are stuck together in a tower and invoke repeatedly their own function 
\textsc{GiveDirection}, then two distinct global directions are given in finite time 
to the two robots regardless of their chirality. This property allows the algorithm
to ``break'' the tower since at least one robot is then able to leave the node where the
tower is located.

\begin{algorithm}
\caption{Function Update} \label{algo:UpdateFunction}
\footnotesize
    \begin{algorithmic} [1]
    \Function{Update}{}
        \If {$ExistsAdjacentEdge()$}
            \State $NumberRobotsPreviousEdgeActivation\leftarrow NumberOfRobotsOnNode()$
            \State $HasMovedPreviousEdgeActivation\leftarrow ExistsEdgeOnCurrentDirection()$
        \EndIf
    \EndFunction
    \end{algorithmic}
\end{algorithm}

Finally, we define the function \textsc{OppositeDirection} that simply affects the value
$left$ (resp. $right$) to the variable $dir$ when $dir=right$ (resp. $dir=left$).

There are two types of configurations in which the robots may change the direction 
they consider. So, our algorithm needs to identify them. We do so by defining a 
predicate that characterizes each of these configurations.

The first one, called $We\-Are\-Stuck\-In\-The\-Same\-Dire\-ction()$, is dedicated to the detection 
of configurations in which the robot must invoke the ``tower breaking'' mechanism. 
Namely, the robot is stuck since at least one edge-activation with at least another
robot and the edge in the direction opposite to the one considered by the robot is
present. More formally, this predicate is defined as follows:

\vspace{0.2cm}
\noindent\begin{tabular}{l@{}l@{}l}
$We$ & $Ar$ & $eStuckInTheSameDirection()\equiv$\\
& & $(NumberOfRobotsOnNode() > 1)$\\
& $\wedge$ & $(NumberOfRobotsOnNode()=NumberRobotsPreviousEdgeActivation)$\\
& $\wedge$ & $\lnot ExistsEdgeOnCurrentDirection()$\\
& $\wedge$ & $ExistsEdgeOnOppositeDirection()$\\
& $\wedge$ & $\lnot HasMovedPreviousEdgeActivation$
\end{tabular}
\vspace{0.1cm}

The second predicate, called $I\-Was\-Stuck\-On\-My\-No\-de\-And\-Now\-We\-Are\-More\-Ro\-bots()$, is 
designed to detect configurations in which the robot must transition from 
the ``sentinel'' to the ``visitor'' role in the ``sentinel''/``visitor'' scheme. 
More precisely, such configuration is characterized by the fact that the robot is 
edge-activated, stuck during its previous edge-activation, and there are strictly
more robots located at its node than at its previous edge-activation. More formally, 
this predicate is defined as follows:

\vspace{0.2cm}
\noindent\begin{tabular}{l@{}l@{}l}
$IW$ & $as$ & $StuckOnMyNodeAndNowWeAreMoreRobots()\equiv$\\
& & $(NumberOfRobotsOnNode()>NumberRobotsPreviousEdgeActivation)$\\
& $\wedge$ & $\lnot HasMovedPreviousEdgeActivation$ \\
& $\wedge$ & $ExistsAdjacentEdge()$
\end{tabular}
\vspace{0.1cm}


\begin{algorithm}
\caption{Function GiveDirection} \label{algo:GiveDirection}
\footnotesize
	\begin{algorithmic} [1]
	\Function{GiveDirection}{}
		\State $i$ $\leftarrow$ $i + 1 \pmod{\ell} + 1$
		\If {$TransformedIdentifier[i] = 0$}
			\State $dir\leftarrow left$
		\Else
			\State $dir\leftarrow right$
		\EndIf
	\EndFunction
	\end{algorithmic}
\end{algorithm}

Now, we are ready to present the pseudo-code of the core of our algorithm 
(see Algorithm \ref{algo:SSPE}). The basic idea of the algorithm is the following. 
The function \textsc{GiveDirection} is invoked when $We\-Are\-Stuck\-In\-The\-Same\-Dire\-ction()$ is true 
(to try to ``break'' the tower after the appropriate waiting),
while the function \textsc{OppositeDirection}
is called when $I\-Was\-Stuck\-On\-My\-Node\-And\-Now\-We\-Are\-More\-Robots()$ 
is true (to implement the ``sentinel''/``visitor'' scheme).  Afterwards, the function 
\textsc{Update} is called (to update the state of the robot according to 
its environment). 

\begin{algorithm}
\caption{\PEF} \label{algo:SSPE}
\footnotesize
    \begin{algorithmic} [1]
            \If {$WeAreStuckInTheSameDirection()$} \label{stuckCondition}
                \State \Call{GiveDirection}{} \label{givedir}
            \EndIf
            \If {$IWasStuckOnMyNodeAndNowWeAreMoreRobots()$} \label{oppositeDirection}
                \State \Call{OppositeDirection}{} \label{changeDirection}
            \EndIf
        \State \Call{Update}{} \label{updateFunction}
    \end{algorithmic}
\end{algorithm}

   \subsection{Preliminaries to the Correctness Proof} 

First, we introduce some definitions and 
preliminary results that are extensively used in the proof.

We saw previously that the notion of tower is central in our algorithm. Intuitively, 
a tower captures the simultaneous presence of all robots of a given set 
on a node at each time of a given interval. We require either the set of robots or
the time interval of each tower to be maximal. Note that the tower is not required 
to be on the same node at each time of the interval (robots of the tower may move
together without leaving the tower). 

We distinguish two kinds of towers according to the agreement of their robots on the global
direction to consider at each time there exists an adjacent edge to their current location
(excluded the last one). If they agreed, the robots form a long-lived tower while they form a
short-lived tower in the contrary case. This implies that a short-lived tower
is broken as soon as the robots forming the tower are edge-activated, while the 
robots of a long-lived tower move together at each edge activation of the tower
(excluded the last one).

\begin{definition}[Tower]
 A tower $T$ is a couple $(S, \theta)$, where $S$ is a set of robots 
 ($|S| > 1$) and $\theta=[t_{s}, t_{e}]$ is an interval of $\mathbb{N}$,
 such that all the robots of $S$ are located at a same node at each instant of time
 $t$ in $\theta$ and $S$ or $\theta$ are maximal for this property.
 Moreover, if the robots of $S$ move during a round 
 $t \in [t_{s}, t_{e}[$, they are required to
 traverse the same edge.
\end{definition}

\begin{definition}[Long-lived tower] \label{LongLivedTower}
 A long-lived tower $T = (S, [t_{s}, t_{e}])$ is a tower such that there is at least one 
 edge-activation of all robots of $S$ in the time interval $[t_{s}, t_{e}[$.
\end{definition}

\begin{definition} [Short-lived tower]
 A short-lived tower $T$ is a tower that is not a long-lived tower.
\end{definition}

For $k > 1$, a long-lived (resp., a short-lived) tower $T=(S, \theta)$ 
with $|S|=k$ is called a $k$-long-lived (resp., a $k$-short-lived) tower.

As there are only three robots on our system, and that in each round each of them consider a 
global direction, we can make the following observation.

\begin{observation}\label{2robotsInTheSameDirection}
There are at least two robots having the same global direction at each instant time.
\end{observation}

In the remainder of this section, we consider an execution $\mathcal{E}$ of 
Algorithm~\ref{algo:SSPE} executed by three robots $r_{1}$, 
$r_{2}$, and $r_{3}$ on a connected-over-time ring $\mathcal{G}$ of size
 $n\in\mathbb{N}$, with $n \geq 4$, starting from an arbitrary configuration.

For the sake of clarity, the value of a variable or a predicate $name$ of
a given robot $r$ after the Look 
phase of a given round $t$ is denoted by the notation $name(r, t)$.

We say that a robot $r$ has a coherent state at time $t$, if the value of its variable 
$Number\-Ro\-bots\-Pre\-vious\-Edge\-Acti\-va\-tion(r, t)$ corresponds to the value of its predicate
$Number\-Of\-Ro\-bots\-On\-No\-de()$ at its previous edge-activation and the 
value of its variable $Has\-Moved\-Pre\-vious\-Edge\-Acti\-va\-tion(r, t)$ corresponds to the value
of its predicate $Exists\-Edge\-On\-Current\-Dire\-ction()$ at its previous edge-activation. The following lemma
states that, for each robot, there exists
a suffix of the execution in which the state of the robot is coherent.

\begin{lemma} \label{convergence}
For any robot, there exists a time from which its state is always coherent.
\end{lemma}

\begin{proof}
 Consider a robot $r$ performing \PEF.
 
 $\mathcal{G}$ belongs to the class of connected-over-time rings, hence at least
 one adjacent edge to each node of $\mathcal{G}$ is infinitely often present.
 This implies that $r$ is infinitely often edge-activated, whatever its 
 location is. Let $t$ be the first time at which $r$ is edge-activated. 
 
 Variables can be updated only during Compute phases of rounds. When executing
 \PEF, the variables $Number\-Ro\-bots\-Pre\-vious\-Edge\-Acti\-va\-tion$ and
 $Has\-Moved\-Pre\-vious\-Edge\-Acti\-va\-tion$ of $r$ are updated with the 
 current values of its predicates $Number\-Of\-Ro\-bots\-On\-No\-de()$ and 
 $Exists\-Edge\-On\-Current\-Dire\-ction()$ only when it is edge-activated. 
 
 Therefore from time $t + 1$, $r$ is in a coherent state.
\end{proof}

Let $t_{1}$, $t_{2}$, and $t_{3}$ be respectively the 
time at which the robots $r_{1}$, $r_{2}$, and $r_{3}$, respectively are in a
coherent state. Let $t_{max} = max\{t_{1}, t_{2}, t_{3}\}$. 
From Lemma~\ref{convergence}, the three robots are in a coherent state from
$t_{max}$. In the remaining of the proof, we focus on the suffix of the execution 
after $t_{max}$. 

The two following lemmas (in combination with Lemma 4.5 and Corollary 4.1) aim
at showing that, regardless of the chirality of the robots 
and the initial values of their variables $i$, a finite number of synchronous 
invocations of the function \textsc{GiveDirection} by two robots of a long-lived tower 
returns them a distinct global direction. This property is shown by looking 
closely to the structure of the binary representation of the
transformed identifiers of the robots.

To state these lemmas, we need to introduce some vocabulary and definitions from
combinatorics on words. We consider words as (possibly infinite) sequence of letters
from the alphabet $A=\{0,1\}$. Given a word $u$, we refer to its $i$-th letter by $u[i]$.
The length of a word $u$ (denoted $|u|$) is its number of letters. Given two words 
$u=u[1]\ldots u[k]$ and $v=v[1]\ldots v[\ell]$ (with $k=|u|$ and $\ell=|v|$),
the concatenation of $u$ and $v$ (denoted $u.v$) is the word $u[1]\ldots u[k]
v[1]\ldots v[\ell]$ (with $|u.v|=k+\ell$). Given a finite word $u$, the word 
$u^1$ is $u$ itself and the word $u^z$ ($z>1$) is the word $u.u^{z-1}$.
Given a finite word $u$, the word $u^\omega$ is the infinite word $u.u.u.\ldots$.
A prefix $u_1$ of a word $u$ is a word such that there exists a word $u_2$ satisfying
$u=u_1.u_2$. A suffix $u_2$ of a word $u$ is a word such that there exists a word 
$u_1$ satisfying $u=u_1.u_2$. A factor $u_2$ of a word $u$ is a word such that there
exists a prefix $u_1$ and a suffix $u_3$ of $u$ satisfying $u=u_1.u_2.u_3$. The 
factor of $u$ starting from the $i^{th}$ bit of $u$ and ending to the $j^{th}$
bit of $u$ included is denoted $u[i \ldots j]$. A circular permutation of a word
$u$ is a word of the form $u_2.u_1$ where $u=u_1.u_2$.

\begin{lemma} \label{separation_possible}
 Let $u$ and $v$ be two distinct transformed identifiers. If $u^{\omega}$ and 
 $v^{\omega}$ share a common factor $X$, then $X$ is finite. 
\end{lemma}

\begin{proof}
 Consider two distinct transformed identifiers $u$ and $v$ such that $u \neq v$.

 By definition, the transformed identifier $u$ is either equal to $00.010$ or to 
 $11.\Pi_{d = 1}^{\alpha(u)}(\Pi_{1}^{\beta(u, d)}00.\Pi_{1}^{\gamma(u,d)}11).010$ $(*)$
 with $\alpha(u)$ a function giving the number of blocks
 $(\Pi_{1}^{\beta(u, d)}00.\Pi_{1}^{\gamma(u,d)}11)$ contained in $u$, 
 $\beta(u, d)$ a function giving the number of pair of bits 00 contained in the 
 $d^{th}$ block of $u$, and $\gamma(u,d)$ a function giving the number of pair 
 of bits 11 contained in the $d^{th}$ block of $u$.

 Similarly, by definition $v$ is either equal to $00.010$ or to 
 $11.\Pi_{d = 1}^{\alpha(v)}(\Pi_{1}^{\beta(v, d)}00.\Pi_{1}^{\gamma(v,d)}11).010$ $(**)$.
 
 Let $U = u^{\omega}$ and $V = v^{\omega}$.

 Assume by contradiction that $U$ and $V$ share a common factor $X$ of infinite 
 size. Hence $U = x.X$ and $V = y.X$, with $x$ (respectively $y$) 
 the prefix of $U$ (respectively of $V$). 
 We have $X = \tilde{u}^{\omega}$, where $\tilde{u}$ is a circular 
 permutation of the word $u$, and $X = \tilde{v}^{\omega}$, where 
 $\tilde{v}$ is a circular permutation of the word $v$.
 
 By definition of a common factor we 
 have $\forall h \in \mathbb{N}^{*}, U[|x| + h] = V[|y| + h]$ $(***)$.
 
 Let $k \in \mathbb{N}^{*}$ such that $k > |x|$ and such that 
 $U[|x| + k] = 0$, $U[|x| + k + 1] = 1$ and $U[|x| + k + 2] = 0$. By $(*)$ and 
 since $U = x.X = x.\tilde{u}^{\omega}$, $k$ exists. By $(*)$ and by 
 construction of $U$, we know that 
 $U[|x| + k + 3 \ldots |x| + k + |u| + 2]$ is equal to $u$ and 
 $U[|x| + k + 3 \ldots |x| + k + |u| - 1]$ is either equal to $00$ or
 to $11.\Pi_{d = 1}^{\alpha(u)}(\Pi_{1}^{\beta(u, d)}00.\Pi_{1}^{\gamma(u,d)}11)$.
 
 By $(***)$, we have $V[|y| + k] = 0$, $V[|y| + k + 1] = 1$ and 
 $V[|y| + k + 2] = 0$. By $(**)$ and by construction of $V$, we know that 
 $V[|y| + k + 3 \ldots |y| + k + |v| + 2]$ is equal to $v$ and
 $V[|y| + k + 3 \ldots |y| + k + |v| - 1]$ is either equal to $00$ or
 to $11.\Pi_{d = 1}^{\alpha(v)}(\Pi_{1}^{\beta(v, d)}00.\Pi_{1}^{\gamma(v,d)}11)$.
 
 \begin{description}
  \item [Case 1:] \textbf{$\mathbf{|u| = |v|}$.}
  
  If $|u| = |v|$, then by $(***)$ we have 
  $U[|x| + k + 3 \ldots |x| + k + |u| + 2] = V[|y| + k + 3 \ldots |y| + k + |v| + 2]$.
  This implies that $u = v$, which leads to a contradiction with the fact that 
  $u$ and $v$ are distinct. 
 
  \item [Case 2:] \textbf{$\mathbf{|u| \neq |v|}$.}
  
  Without lost of generality assume that $|u| < |v|$.
  We have $U[|x| + k + |u|] = 0$, $U[|x| + k + |u| + 1] = 1$ and 
  $U[|x| + k + |u| + 2] = 0$. Therefore by $(***)$ we have 
  $V[|y| + k + |u|] = 0$, $V[|y| + k + |u| + 1] = 1$ and 
  $V[|y| + k + |u| + 2] = 0$.   
  
  Note that $|u| = 2w + 3$ with $w \in \mathbb{N}^{*}$. Similarly 
  $|v| = 2z + 3$, with $z \in \mathbb{N}^{*}$, and $z > x$ since $|u| < |v|$.
  Since $V[|y| + k + 3 \ldots |y| + k + |v| + 2]$ is equal to $v$, this implies
  that $V[|y| + k + 3] = v[1]$, and 
  $V[|y| + k + |u|] = V[|y| + k + 2w + 3] = v[i]$ where $i$ is odd and such that 
  $1 \leq i \leq 2z$. Hence by $(**)$, necessarily 
  $V[|y| + k + |u|] = V[|y| + k + |u| + 1]$, which leads to a contradiction with
  the fact that $V[|y| + k + |u|] = 0$ and $V[|y| + k + |u| + 1] = 1$.
 \end{description} 
\end{proof}

Let us introduce the notation $\overline{w}$ which given a word
$w$ is defined such that 
$\overline{w} = \prod_{i \in \{0, \ldots, |w| - 1\}} \overline{w[i]}$
where if $w[i] = 1$ then $\overline{w[i]} = 0$, and if $w[i] = 0$ then 
$\overline{w[i]} = 1$.

\begin{lemma} \label{separation_possible_bis}
 Let $u$ and $v$ be two distinct transformed identifiers. If $u^{\omega}$ and
 $\overline{v}^{\omega}$ share a common factor $X$, then $X$ is finite. 
\end{lemma}

\begin{proof}
 Consider two distinct transformed identifiers $u$ and $v$ such that $u \neq v$.

 By definition, the transformed identifier $u$ is either equal to $00.010$ or to 
 $11.\Pi_{d = 1}^{\alpha(u)}(\Pi_{1}^{\beta(u, d)}00.\Pi_{1}^{\gamma(u,d)}11).010$ $(*)$
 with $\alpha(u)$ a function giving the number of blocks
 $(\Pi_{1}^{\beta(u, d)}00.\Pi_{1}^{\gamma(u,d)}11)$ contained in $u$, 
 $\beta(u, d)$ a function giving the number of pair of bits 00 contained in the 
 $d^{th}$ block of $u$, and $\gamma(u,d)$ a function giving the number of pair 
 of bits 11 contained in the $d^{th}$ block of $u$.

 Similarly, by definition $v$ is either equal to $00.010$ or to 
 $11.\Pi_{d = 1}^{\alpha(v)}(\Pi_{1}^{\beta(v, d)}00.\Pi_{1}^{\gamma(v,d)}11).010$.
 Call $w = \overline{v}$. This implies that $|w| = |v|$ and $w$ is either equal
 to $11.101$ or to 
 $00.\Pi_{d = 1}^{\alpha(v)}(\Pi_{1}^{\beta(v, d)}11.\Pi_{1}^{\gamma(v,d)}00).101$ $(**)$.
 Note that $u$ and $w$ are distinct. Indeed, if $|u| \neq |v|$ then, $w$ and 
 $u$ are distinct since $|w| = |v|$. If $|u| = |v|$
 then, since the suffix of size $3$ of $u$ is the word $010$, and the suffix of
 size $3$ of $w$ is the word $101$, then $u$ and $w$ are distinct.
 
 Let $U = u^{\omega}$ and $W = w^{\omega}$.

 Assume by contradiction that $U$ and $W$ share a common factor $X$ of infinite
 size. Hence $U = x.X$ and $W = y.X$, with $x$ (respectively $y$) 
 the prefix of $U$ (respectively of $W$). 
 We have $X = \tilde{u}^{\omega}$, where $\tilde{u}$ is a circular 
 permutation of the word $u$, and $X = \tilde{w}^{\omega}$, where 
 $\tilde{w}$ is a circular permutation of the word $w$.

 By definition of a common factor we 
 have $\forall h \in \mathbb{N}^{*}, U[|x| + h] = W[|y| + h]$ $(***)$.
 
 Let $k \in \mathbb{N}^{*}$ such that $k > |x|$ and such that 
 $U[|x| + k] = 0$, $U[|x| + k + 1] = 1$ and $U[|x| + k + 2] = 0$. By $(*)$ and 
 since $U = x.X = x.\tilde{u}^{\omega}$, $k$ exists. By $(*)$ and by 
 construction of $U$, we know that 
 $U[|x| + k + 3 \ldots |x| + k + |u| + 2]$ is equal to $u$.
 
 By $(***)$, we have $W[|y| + k] = 0$, $W[|y| + k + 1] = 1$ and 
 $W[|y| + k + 2] = 0$. By $(**)$ and by construction of $W$, we know that 
 either $W[|y| + k + 4 \ldots |y| + k + |w| + 3] = w$ (in the case where 
 $W[|y| + k + 1] = w[|w| - 2]$) or 
 $W[|y| + k + 2 \ldots |y| + k + |w| + 1] = w$ (in the case where 
 $W[|y| + k + 1] = w[|w|]$).
 
 \begin{description}

  \item [Case 1:] \textbf{$\mathbf{W[|y| + k + 4 \ldots |y| + k + |w| + 3] = w}$.}
     
   In this case $W[|y| + k + 3] = 1$, then necessarily by $(***)$ 
   $U[|x| + k + 3] = 1$. By $(*)$, and since 
   $U[|x| + k + 3 \ldots |x| + k + |u| + 2] = u$, this implies that 
   $U[|x| + k + 4] = 1$. Therefore by $(***)$, necessarily $W[|y| + k + 4] = 1$.
   Since $W[|y| + k + 4 \ldots |y| + k + |w| + 3] = w$, and by $(**)$,
   this implies that $w = 11.101$, otherwise $W[|y| + k + 4] = 0$, which leads
   to a contradiction with the fact that $U[|x| + k + 4] = 1$.
   
   This implies by $(***)$, that $U[|x| + k + 3 \ldots |x| + k + 8] = 111101$.
   Therefore by $(*)$, necessarily $U[|x| + k + 9] = 0$. However by construction
   of $W$, since $W[|y| + k + 4 \ldots |y| + k + |w| + 3] = w$, and since 
   $|w| = 5$, we have $W[|y| + k + 9 \ldots |y| + k + 14] = w$. This implies
   that $W[|y| + k + 9] = 1$ since $w = 11.101$, which leads to a contradiction
   with the fact that $U[|x| + k + 9] = 0$.

  \item [Case 2:] \textbf{$\mathbf{W[|y| + k + 2 \ldots |y| + k + |w| + 1] = w}$.}
  
  In this case, since $W[|y| + k + 2] = 0$, this implies by $(**)$ that 
  $W[|y| + k + 3] = 0$. Therefore by $(***)$ we have $U[|x| + k + 3] = 0$. 
  Hence, since $U[|x| + k + 3 \ldots |x| + k + |u| + 2] = u$, then by $(*)$,
  we have $u = 00.010$, otherwise $U[|x| + k + 3] = 1$ which leads to a
  contradiction with the fact that $W[|y| + k + 3] = 0$.
 
  This implies by $(***)$, that $W[|y| + k + 2 \ldots |x| + k + 7] = 000010$.
  Therefore by $(**)$, necessarily $W[|x| + k + 8] = 1$. However by construction
  of $U$, since $U[|x| + k + 3 \ldots |x| + k + |u| + 2] = u$, and since 
  $|u| = 5$, we have $U[|x| + k + 8 \ldots |y| + k + 13] = u$ which implies that 
  $U[|x| + k + 8] = 0$ since $u = 00.010$, which leads to a contradiction with
  the fact that $W[|y| + k + 8] = 1$.
 \end{description}
\end{proof}

\subsection{Tower Properties} 

We are now able to state a set of lemmas
that show some interesting technical properties of towers under specific assumptions 
during the execution of our algorithm.
These properties are extensively used in the main proof of our algorithm.

\begin{lemma} \label{propertyLongLivedTower}
 The robots of a long-lived tower $T = (S, [t_{s}, t_{e}])$ consider a same 
 global direction at each time between the Look phase of 
 round $t_{s}$ and the Look phase of round $t_{e}$ included.
\end{lemma}

\begin{proof}
 Consider a long-lived tower $T = (S, [t_{s}, t_{e}])$.
 
 Call $t_{act}$ the first time in $[t_{s}, t_{e}[$ at which the robots of $S$ 
 are edge-activated. Since $T$ is a long-lived tower, $t_{act}$ exists.
 
 When executing \PEF, a robot can change the global direction it considers only
 when it is edge-activated. Moreover a robot does not change the global
 direction it considers if it has moved during its previous edge-activation. 
 Besides, during the Look phase of a time $t$ a robot considers the same global
 direction than the one it considers during the Move phase of time $t - 1$.

 Therefore, during the Look phase of time $t_{s}$ the robots of $S$ consider the
 same global direction, otherwise the robots of $S$ consider different global
 directions during the Move phase of time $t_{s} - 1$, and so move during this 
 phase (otherwise $T$ is not formed at time $t_{s}$), therefore they separate
 during the Move phase of time $t_{act}$. This 
 leads to a contradiction with the fact that $T$ is a long-lived tower.

 Consider a time $t \in ]t_{s}, t_{e}[$. If at time $t$ the robots of $S$ are 
 not edge-activated, then during the Move phase of time $t$ the robots of $S$ do
 not change the global direction they consider. 
 
 $T$ is a long-lived tower from time $t_{s}$ to time $t_{e}$ included. Therefore
 if at time $t \in ]t_{s}, t_{e}[$ the robots of $S$ are edge-activated, then, 
 by definition of a long-lived tower, during the Move phase of time $t$, the 
 robots of $S$ consider the same global direction. 
 
 
 Since at time $t_{s}$ the robots of $S$ consider the same global direction 
 using the two previous arguments by recurrence on each time 
 $t \in ]t_{s}, t_{e}[$ and the fact that robots change the global directions
 they consider only during Compute phases, we can conclude that the robots of 
 $S$ consider a same global direction from the Look phase of time $t_{s}$ to the
 Look phase of time $t_{e}$.
\end{proof}

The following lemma is used to prove, in combination with Lemmas 
\ref{separation_possible} and \ref{separation_possible_bis}, the 
``tower breaking'' mechanism
since it proves that robots of a long-lived tower synchronously invoke their
\textsc{GiveDirection} function after their first edge-activation.

\begin{lemma} \label{successive_bits_consideration}
 For any long-lived tower $T = (S, [t_{s}, t_{e}])$, any 
 $(r_{i}, r_{j})$ in $S^{2}$, and any $t$ less or equal to $t_{e}$, we have
 $We\-Are\-Stuck\-In\-The\-Same\-Dire\-ction()(r_{i}, t)$ $=$ 
 $We\-Are\-Stuck\-In\-The\-Same\-Dire\-ction()(r_{j}, t)$ and 
 $I\-Was\-Stuck\-On\-My\-No\-de\-And\-Now\-We\-Are\-More\-Ro\-bots()(r_{i}, t)$
 $=$ 
 $I\-Was\-Stuck\-On\-My\-No\-de\-And\-Now\-We\-Are\-More\-Ro\-bots()(r_{j}, t)$
 if all robots of $S$ have been edge-activated between $t_{s}$ (included) and 
 $t$ (not included).
\end{lemma}

\begin{proof}
 Consider a long-lived tower $T = (S, [t_{s}, t_{e}])$. Let $t_{act}$ be the
 first time in $[t_{s}, t_{e}[$ where the robots of $S$ are edge-activated. By 
 definition of a long-lived tower, this time exists.
 
 By definition of a long-lived tower and by
 lemma~\ref{propertyLongLivedTower}, from the Look phase of time $t_{s}$ to the
 Look phase of time $t_{e}$ included, all the robots of $S$ are on a same node 
 and consider a same global direction. Therefore the values of their respective 
 predicates $Number\-Of\-Ro\-bots\-On\-No\-de()$, 
 $Exists\-Edge\-On\-Current\-Dire\-ction()$, 
 $Exists\-Edge\-On\-Oppo\-si\-te\-Dire\-ction()$ and 
 $Exists\-Adja\-cent\-Edge()$ are identical from the Look phase of time $t_{s}$ 
 to the Look phase of time $t_{e}$ included.
 
 When executing \PEF, a robot updates its variables 
 $Number\-Ro\-bots\-Pre\-vious\-Edge\-Acti\-va\-tion$ and 
 $Has\-Moved\-Pre\-vious\-Edge\-Acti\-va\-tion$ respectively with the values of
 its predicates $Number\-Of\-Ro\-bots\-On\-No\-de()$ and 
 $Exists\-Edge\-On\-Current\-Dire\-ction()$, only during Compute phases of times 
 where it is edge-activated. By the observation made at the previous paragraph, 
 this implies that from the Compute phase of time $t_{act}$ to the Look phase of
 time $t_{e}$ included, the robots of $S$ have the same values for their
 variables $Number\-Ro\-bots\-Pre\-vious\-Edge\-Acti\-va\-tion$ and 
 $Has\-Moved\-Pre\-vious\-Edge\-Acti\-va\-tion$.
 
 Then, by construction of the predicates
 $We\-Are\-Stuck\-In\-The\-Same\-Dire\-ction()$ and 
 $I\-Was\-Stuck\-On\-My\-No\-de\-And\-Now\-We\-Are\-More\-Ro\-bots()$, the lemma
 is proved.

\end{proof}

From the Lemmas~\ref{successive_bits_consideration}, \ref{separation_possible} and
\ref{separation_possible_bis}, we can then deduce the following corollary.

\begin{corollary} \label{breakingTowerCorollary}
 Consider a long-lived tower $T = (S, \theta)$ with $\theta = [t_{s}, +\infty[$.
 The predicates $We\-Are\-Stuck\-In\-The\-Same\-Dire\-ction()$ of the robots of
 $S$ cannot be infinitely often true, otherwise $T$ is broken in finite time.
\end{corollary}

\begin{proof}
 First, note that if two robots possess two distinct identifiers, then their
 transformed identifiers are also distinct.

 Consider a long-lived tower $T = (S, \theta)$ with $\theta = [t_{s}, +\infty[$.
 
 Call $t_{act} \geq t_{s}$ the first time after $t_{s}$ where the robots of $S$ 
 are edge-activated. By definition of a long-lived tower, $t_{act}$ exists. By 
 Lemma~\ref{successive_bits_consideration}, after time $t_{act}$, the robots
 of $S$ consider the same values of predicates 
 $We\-Are\-Stuck\-In\-The\-Same\-Dire\-ction()$ 
 and $I\-Was\-Stuck\-On\-My\-No\-de\-And\-Now\-We\-Are\-More\-Ro\-bots()$. 
 
 Assume by contradiction that after $t_{act}$ the predicates 
 $We\-Are\-Stuck\-In\-The\-Same\-Dire\-ction()$ of the robots of $S$ are
 infinitely often true. Then by construction of \PEF, after time $t_{act}$, all
 the robots of $S$ call the function \textsc{GiveDirection} infinitely often and
 at the same instants of times. 
 
 If among the robots of $S$ two have the same chirality, to keep forming $T$ 
 they need to consider the same values of bits each time the 
 function \textsc{GiveDirection} is called. Here the robots have to consider the 
 same values of bits infinitely often (since the two robots call the function 
 \textsc{GiveDirection} infinitely often). Each time a robot executes the 
 function \textsc{GiveDirection} it reads the next bit (in a round robin way) of
 the bit read during its previous call to the function \textsc{GiveDirection}.
 Call $i_{1}$ and $i_{2}$ the two respective transformed identifiers of two 
 robots forming $T$ such that these two robots possess the same chirality. By
 the previous observations, to keep forming $T$, $i_{1}^{\omega}$ and
 $i_{2}^{\omega}$ must share an infinite common factor. However according to 
 Lemma~\ref{separation_possible} this is not possible. Therefore there exists a 
 time $t_{end}$ at which these two robots consider two different bits. When the
 robots call the function \textsc{GiveDirection}, they are edge-activated (by 
 definition of the predicate $We\-Are\-Stuck\-In\-The\-Same\-Dire\-ction()$), 
 therefore at time $t_{end}$, $T$ is broken.

 Similarly, if among the robots of $S$ two have not the same chirality, to keep
 forming $T$ they need to consider different values of bits each time the
 function \textsc{GiveDirection} is called. Here the robots have to consider
 different values of bits infinitely often (since the two robots call the 
 function \textsc{GiveDirection} infinitely often). Each time a robot executes
 the function \textsc{GiveDirection} it reads the next bit (in a round robin 
 way) of the bit read during its previous call to the function 
 \textsc{GiveDirection}. Call $j_{1}$ and $j_{2}$ the two respective transformed
 identifiers of two robots forming $T$ such that these two robots possess a
 different chirality. By the previous observations, to keep forming $T$, 
 $j_{1}^{\omega}$ must possess an infinite suffix $S$ such that an infinite 
 suffix of $j_{2}^{\omega}$ is equal to $\overline{S}$. This is equivalent to 
 say that $j_{1}^{\omega}$ and $\overline{j_{2}^{\omega}}$ must possess an 
 infinite common factor. However according to
 Lemma~\ref{separation_possible_bis} this is not possible.
 Therefore there exists a time $t_{end}$ at which these two robots consider two
 identical bits. When the robots call the function \textsc{GiveDirection}, they 
 are edge-activated, therefore at time $t_{end}$, $T$ is broken.
 
 Hence in both cases the long-lived tower $T$ is broken, which leads to a
 contradiction with the fact that $\theta = [t_{s}, +\infty[$.
\end{proof}

\begin{lemma} \label{OneEdgeMissingTowerBroken}
 If there exists an eventual missing edge, then all long-lived towers have a 
 finite duration.
\end{lemma}

\begin{proof}
 Consider that there exists an edge $e$ of $\mathcal{G}$ which is missing 
 forever from time $t_{missing}$. Consider the execution from time $t_{missing}$.
 
 Call $u$ and $v$ the two adjacent nodes of $e$, such that $v$ is the adjacent 
 node of $u$ in the clockwise direction.
 
 By contradiction assume that there exists a long-lived tower $T = (S, \theta)$
 such that $\theta = [t_{s}, +\infty[$. Exactly 3 robots are executing \PEF, 
 so $\textpipe S \textpipe$ is either equals to 2 or 3. We want to prove that 
 all the robots of $T$ have their predicates 
 $We\-Are\-Stuck\-In\-The\-Same\-Dire\-ction()$ infinitely often true. By 
 contradiction, assume that there exists a robot $r_{i}$ of $S$, such that it 
 exists a time $t_{i}$ in $\theta$ such that for all time $t$ greater or equal 
 to $t_{i}$ its predicate $We\-Are\-Stuck\-In\-The\-Same\-Dire\-ction()$ is 
 false.
 
 Call $t_{act} \geq t_{s}$, the first time after time $t_{s}$ where the robots are 
 edge-activated. Since $T$ is a long-lived tower, $t_{act}$ exists. By 
 Lemma~\ref{successive_bits_consideration}, from time
 $t_{act} + 1$ the robots of $S$ possess the same values of predicates 
 $We\-Are\-Stuck\-In\-The\-Same\-Dire\-ction()$. By assumption of contradiction, from time 
 $t_{false} = max \{t_{act} + 1, t_{i}\}$ the predicates
 $We\-Are\-Stuck\-In\-The\-Same\-Dire\-ction()$ of all the robots of $S$ are 
 false.
 
 We recall that by definition of a long-lived tower and by 
 Lemma~\ref{propertyLongLivedTower} all the robots of $S$ are on a same node and
 consider a same global direction from the Look phase of time $t_{s}$ to the 
 Look phase of time $t_{e}$ included.
 
 \begin{description}
  \item [Case 1: $\mathbf{\textpipe S \textpipe}$ $=$ $\mathbf{3}$.]
  
  From time $t_{s}$ the predicates $Number\-Of\-Ro\-bots\-On\-No\-de()$ of the
  robots of $S$ are equal to 3. When executing \PEF, a robot updates its 
  variables $Number\-Ro\-bots\-Pre\-vious\-Edge\-Acti\-va\-tion$ with the value
  of its predicate $Number\-Of\-Ro\-bots\-On\-No\-de()$, only during Compute 
  phases of times where it is edge-activated. Therefore from time $t_{false}$ 
  the robots of $S$ have their variables 
  $Number\-Ro\-bots\-Pre\-vious\-Edge\-Acti\-va\-tion$ equal to 3. Hence, from 
  time $t_{false}$ their predicates 
  $I\-Was\-Stuck\-On\-My\-No\-de\-And\-Now\-We\-Are\-More\-Ro\-bots()$ are 
  false, since the condition
  $Number\-Of\-Ro\-bots\-On\-No\-de() > Number\-Ro\-bots\-Pre\-vious\-Edge\-Acti\-va\-tion$
  is false. 
  
  Since from time $t_{false}$, the predicates 
  $We\-Are\-Stuck\-In\-The\-Same\-Dire\-ction()$ of the robots of $S$ are also 
  false, then from time $t_{false}$ the robots of $S$ always consider the same
  global direction.
 
  Without lost of generality, assume that, from time $t_{false}$, the robots of 
  $S$ consider the clockwise direction. All the edges of $\mathcal{G}$ except 
  $e$ are infinitely often present, therefore the robots of $S$ reach node $u$ 
  in finite time. However $e$ is missing forever, hence in finite time, the 
  predicates $We\-Are\-Stuck\-In\-The\-Same\-Dire\-ction()$ of all the robots 
  are true. This leads to a contradiction.
  
  \item [Case 2: $\mathbf{\textpipe S \textpipe}$ $\mathbf{=}$ $\mathbf{2}$.]
  
  Assume, without lost of generality, that $T$ is formed of $r_{1}$ and $r_{2}$.

  If, after $t_{false}$, the 2-long-lived tower does not meet $\mathbf{r_{3}}$, 
  then by similar arguments than the one used for the case 1 we prove that there
  is a contradiction.
   
  Now consider the case where the 2-long-lived tower meets $\mathbf{r_{3}}$.       
   If at a time $t' > t_{false}$, the robots of $S$ meet $r_{3}$ it is either 
   because the two entities (the tower and $r_{3}$) move during the Move phase 
   of time $t' - 1$ while considering two opposed global directions or because
   the two entities consider the same global direction but one of the entity
   cannot move (an edge is missing in its direction) during the Move phase 
   of round $t' - 1$. 
   Let $t_{act}' \geq t'$ be the first time after time $t'$ included where the
   three robots are edge-activated. All the edges of $\mathcal{G}$ except $e$ 
   are infinitely often present therefore $t_{act}'$ exists. In both cases,
   thanks to the update at time $t' - 1$ of the variables
   $Has\-Moved\-Pre\-vious\-Edge\-Acti\-va\-tion$ and 
   $Number\-Ro\-bots\-Pre\-vious\-Edge\-Acti\-va\-tion$ of the robots, during 
   the Move phase of time $t_{act}'$ the robots of the two entities consider 
   opposed global directions. 
   The two entities separate them during the Move phase of this 
   time. Moreover, from this separation, as long as $r_{3}$ is alone on its node
   it does not change the global direction it considers. Similarly, from this
   separation, as long as the robots of $S$ do not meet $r_{3}$, their 
   predicates
   $I\-Was\-Stuck\-On\-My\-No\-de\-And\-Now\-We\-Are\-More\-Ro\-bots()$ are 
   false, and since from time $t_{false}$ their predicates 
   $We\-Are\-Stuck\-In\-The\-Same\-Dire\-ction()$ are false, they do not change 
   the global direction they consider. 
   
   Hence, in finite time after time $t_{act}'$ the two entities are located 
   respectively on the two extremities of $e$. However $e$ is missing forever, 
   therefore in finite time, the predicates 
   $We\-Are\-Stuck\-In\-The\-Same\-Dire\-ction()$ of the robots of $T$ are true.
   This leads to a contradiction.
%
 \end{description}
 In both cases a contradiction is highlighted. Therefore, after $t_{false}$ all
 the robots of $S$ have their predicates 
 $We\-Are\-Stuck\-In\-The\-Same\-Dire\-ction()$ infinitely often true. Then we
 can use Corollary~\ref{breakingTowerCorollary} to prove that $T$ is broken, 
 which leads to a contradiction with the fact that $\theta = [t_{s}, +\infty[$.
\end{proof}

\begin{lemma} \label{noThreeShortLivedTower}
 No execution containing only configurations without long-lived tower reaches a
 configuration where three robots form a tower.
\end{lemma}

\begin{proof}
 Assume that there is no long-lived tower in the execution. The robots can cross
 at most one edge at each round. Each node has at most 2 adjacent edges in 
 $\mathcal{G}$. Moreover each robot considers at each instant time a direction.
 Assume, by contradiction that 3 robots form a tower $T$ at a time $t$. Let 
 $t' \geq t$ be the first time after time $t$ where the robots of $T$ are
 edge-activated. There is no 3-long-lived tower in the execution, therefore 
 during the Move phase of time $t'$, the robots of $T$ consider two opposed 
 global directions. However there are three robots, and two different global
 directions, hence, during the Move phase of time $t'$, two robots consider the
 same global direction. Therefore there exists a 2-long-lived tower, which leads
 to a contradiction.
\end{proof}

\begin{lemma} \label{leLemmeQuilFallait}
 In every execution, if a tower involving 3 robots is formed at time $t$, then
 at time $t - 1$ a 2-long-lived tower is present in $\varepsilon$. 
\end{lemma}

\begin{proof}
 Assume that a tower $T$ of 3 robots is formed at time $t$.
 
 First note that if there exists a 2-long-lived tower $T' = (S, [t_{s}, t_{e}])$
 such that $t - 1 \in [t_{s}, t_{e}[$, it is possible for $T$ to be formed.
 
 Now we prove that if there is no 2-long-lived tower at time $t- 1$ then $T$ 
 cannot be formed at time $t$. Assume that at time $t - 1$ there is no 
 2-long-lived tower. Let us consider the three following cases.
 
 \begin{description}
  \item [Case 1:] \textbf{There is a tower $\mathbf{T'}$ of 3 robots at time 
  $\mathbf{t - 1}$.}
  The tower $T'$ must break at time $t - 1$, otherwise there is a contradiction 
  with the fact that $T$ is formed at time $t$. Hence the robots of $T'$ are 
  edge-activated at time $t - 1$. While executing \PEF~the robots consider a 
  direction at each round. There are only two possible directions. Therefore, 
  for the tower $T'$ to break at time $t - 1$, two robots of $T'$ consider a
  same global direction, while the other robot of $T'$ considers the opposite 
  global direction. This implies that the three robots cannot be present on a 
  same node at time $t$, since $n \geq 4$.
 
  \item [Case 2:] \textbf{There is a 2-short-lived tower $T'$ at time 
  $\mathbf{t - 1}$.}
  For the three robots to form $T$ at time $t$, they must be edge-activated at 
  time $t - 1$. By definition of a 2-short-lived tower, the two robots of $T'$ 
  consider two opposed global directions during the Move phase of time $t - 1$.
  Since the robots can cross at most one edge at each round, it is not possible
  for the three robots to be on a same node at time $t$, which leads to a 
  contradiction with the fact that $T$ is formed at time $t$.
 
  \item [Case 3:] \textbf{There are 3 isolated robots at time $\mathbf{t - 1}$.}
  For the three robots to form $T$ at time $t$, they must be edge-activated at 
  time $t - 1$. The robots can cross at most one edge at each round. Each node
  has at most 2 adjacent edges present in $\mathcal{G}$. Moreover each robot considers
  at each instant time a direction. Therefore it is not possible for the three
  robots to be on a same node at time $t$, which leads to a contradiction with 
  the fact that $T$ is formed at time $t$.
 \end{description}
\end{proof}

\begin{lemma} \label{noThreeLongLivedTower}
 Every execution starting from a configuration without a 3-long-lived tower
 cannot reach a configuration with a 3-long-lived tower.
\end{lemma}

\begin{proof}
 Assume that $\mathcal{E}$ starts from a configuration which does not contain a
 3-long-lived tower. By contradiction, let $\gamma$ be the first configuration
 of $\mathcal{E}$ containing a 3-long-lived tower $T = (S, [t_{s}, t_{e}])$.
 
 Let $t_{act} \geq t_{s}$ be the first time after time $t_{s}$ where the 3 
 robots of $T$ are edge-activated. By definition of a long-lived tower, 
 $t_{act}$ exists.
 
 Lemma~\ref{leLemmeQuilFallait} implies that the configuration at time 
 $t_{s} - 1$ contains a 2-long-lived tower. Hence, since $\gamma$ contains the first 
 3-long-lived tower of $\mathcal{E}$, at time $t_{s}$ a 2-long-lived tower and a 
 robot meet to form $T$. The meeting between these two entities can happen
 either because both of them move in opposed global directions during the Move
 phase of time $t_{s} - 1$, or because, during the Move phase of time
 $t_{s} - 1$, the two entities consider the same global direction but one of the
 entity cannot move (an edge is missing in its direction). In both cases; thanks
 to the update of the variables $Has\-Moved\-Pre\-vious\-Edge\-Acti\-va\-tion$
 and $Number\-Ro\-bots\-Pre\-vious\-Edge\-Acti\-va\-tion$
 at time $t_{s} - 1$; during the Move phase of time $t_{act}$ the two entities
 consider opposed global directions. Hence, the two entities separate during the
 Move phase of time $t_{act}$, therefore there is a contradiction with the fact
 that $T$ is a 3-long-live tower.
\end{proof}

\begin{lemma} \label{NoThreeRobotsInSameDirection}
Let $\gamma$ be a configuration such that all but one robots consider the same 
global direction. Then starting from $\gamma$, no execution without
long-lived tower can reach a configuration where all robots consider the same
global direction.
\end{lemma}

\begin{proof}
 Assume that $\mathcal{E}$ does not contain long-lived tower and starts from a 
 configuration where all robots but one consider the same global direction. For
 the three robots to consider the same global direction at least one robot must 
 change the global direction it considers. While executing \PEF, the only way
 for a robot to change the global direction it considers is to form a tower and 
 to be edge-activated. By Lemma~\ref{noThreeShortLivedTower}, there is no tower
 of 3 robots in $\mathcal{E}$. Therefore, for at least one robot to change the
 global direction it considers, a 2-short-lived tower must be formed and the 
 robots of this tower must be edge-activated. However, by definition of a
 2-short-lived tower, once edge-activated, the two robots composing the 
 2-short-lived tower consider two opposed global directions. Therefore after 
 the edge-activation, the three robots do not consider the same global 
 direction.
\end{proof}

\begin{lemma} \label{separationSecondTower2Robots}
 Consider an execution containing no 3-long-lived tower. If a 2-long-lived tower
 $T = (S, [t_{s}, t_{e}])$, where $t_{e}$ is finite, is located at a node $u$ at
 round $t_{e}$, then the robot that does not belong to $S$ cannot be located at
 node $u$ during the Look phase of round $t_{e}$. Moreover, during the Look 
 phase of round $t_{e} + 1$, one robot of $S$ located at $u$ considers a global
 direction opposite to the one considered by the other robot of $S$ (which is
 no longer on $u$).
\end{lemma}

\begin{proof}
 Assume that $\mathcal{E}$ does not contain 3-long-lived tower. Assume that 
 $r_{1}$ and $r_{2}$ are involved in a 2-long-lived tower
 $T = (S, [t_{s}, t_{e}])$. 
 
 After the Compute phase of time $t_{e}$, $r_{1}$ and $r_{2}$ consider two
 opposed global directions, otherwise there is a contradiction with the fact 
 that $T$ is broken at time $t_{e}$. Directions of robots can be modified only 
 during Compute phases of rounds, therefore during the Look phase of time 
 $t_{e} + 1$, the robots of $T$ still consider two opposed global directions.
 
 Let $t_{act} \in [t_{s}, t_{e}[$ be the first time after time $t_{s}$ where the
 robots of $T$ are edge-activated. By definition of a 2-long-lived tower 
 $t_{act}$ exists. By Lemma~\ref{successive_bits_consideration}, from the Look
 phase of time $t_{act} + 1$ to the Look phase of time $t_{e}$ included, $r_{1}$
 and $r_{2}$ have the same values of predicates
 $We\-Are\-Stuck\-In\-The\-Same\-Dire\-ction()$ and  
 $I\-Was\-Stuck\-On\-My\-No\-de\-And\-Now\-We\-Are\-More\-Ro\-bots()$. 
 Therefore, while executing \PEF, the only way for $r_{1}$ and $r_{2}$ to 
 consider two opposed global directions during the Move phase of time $t_{e}$ is
 to execute the function \textsc{GiveDirection} and hence to have their 
 predicates $We\-Are\-Stuck\-In\-The\-Same\-Dire\-ction()$ true. Therefore, at
 time $t_{e}$ the condition $\lnot Exists\-Edge\-On\-Current\-Dire\-ction()
 \wedge Exists\-Edge\-On\-Oppo\-si\-te\-Dire\-ction()$ is true. Hence, since 
 during the Move phase of time $t_{e}$, $r_{1}$ and $r_{2}$ consider two opposed
 global directions, during the Look phase of time $t_{e} + 1$, one of the robot
 of $T$ is still on node $u$, while the other robot of $T$ is on an adjacent 
 node of $u$.
 
 Now assume, by contradiction that $r_{3}$ is on node $u$ during the Look phase
 of time $t_{e}$. Let $t_{last}$ be the last time in $[t_{s}, t_{e}[$, where the
 robots of $T$ are edge-activated. By definition of a 2-long-lived tower 
 $t_{last}$ exists. There is no 3-long-lived tower, hence if $r_{3}$ is on node 
 $u$ at time $t_{e}$, it forms a 3-short-lived tower with the robots of $T$ at
 time $t_{last} + 1$. 
 
 Note that at time $t_{last}$, $r_{3}$ cannot be located on the same node as the 
 robots of $T$, otherwise since $n \geq 4$ the three robots cannot form a 
 3-short-lived tower at time $t_{last} + 1$. This implies that at time 
 $t_{last}$ the function \textsc{Update}, updates the variables 
 $Number\-Ro\-bots\-Pre\-vious\-Edge\-Acti\-va\-tion$ of the robots of $T$ to 
 $2$. Since the variables are updated only during the Compute phases of times 
 where the robots are edge-activated, during the Look phase of time $t_{e}$, the
 robots of $T$ have their variables 
 $Number\-Ro\-bots\-Pre\-vious\-Edge\-Acti\-va\-tion$ still equal to 2. Since 
 $r_{3}$ is on node $u$ during the Look phase of time $t_{e}$, the predicates
 $Number\-Of\-Ro\-bots\-On\-No\-de()$ of the robots of $T$ are not equal to 
 their variables $Number\-Ro\-bots\-Pre\-vious\-Edge\-Acti\-va\-tion$. Therefore
 the robots of $T$ cannot execute the function \textsc{GiveDirection} at time
 $t_{e}$, and hence are not able to separate them, which leads to a 
 contradiction with the fact that $T$ is broken at time $t_{e}$.
\end{proof}

\begin{lemma} \label{existsTower2Robots}
 Consider an execution $\mathcal{E}$ without any 3-long-lived tower. If a 
 2-long-lived tower $T$ is formed at a time $t_{s}$, then during the Look 
 phase of time $t_{s} - 1$, a tower $T'$ of 2 robots involving only one robot of
 $T$ is present. Moreover, during the Move phase of time $t_{s} - 1$, the robot
 of $T$ involved in $T'$ does not move while the other robot of $T$ moves.
\end{lemma}

\begin{proof}
 Consider an execution $\mathcal{E}$ without any 3-long-lived tower. Assume that 
 at time $t_{s}$ a 2-long-lived tower $T = (S, [t_{s}, t_{e}])$ is formed.
 
 First note that if there exists a tower $T'$ of 2 robots at time $t_{s} - 1$, 
 such that only one robot of $T'$ is involved in $T$ and such that this robot 
 does not move during the Move phase of time $t_{s} - 1$, then it is possible
 for $T$ to be formed. Now we prove that $T$ can be formed at time $t_{s}$ only
 in this situation.
 
 Assume, by contradiction, that there is no tower of 2 robots during the
 Look phase of time $t_{s} - 1$. This implies that, at time $t_{s} - 1$ either 
 the three robots are involved in a 3-short-lived tower $T_{3}$ (case 1) or the 
 three robots are isolated (case 2).
  
 \begin{description}
  \item [Case 1:]
  
  Call $t$, the time of the formation of $T_{3}$.
  By Lemma~\ref{leLemmeQuilFallait}, at time $t - 1$, there is a
  2-long-lived tower $T" = (S", \theta")$ in $\mathcal{E}$ such that
  $t \in \theta"$. Call $r$ the robot that does not belong to $T"$. 
  Note that $S" \neq S$, otherwise there is a contradiction with the fact that
  $T$ starts at time $t_{s}$. This implies that $T$ is composed of one robot of 
  $T"$ and of $r$. However by lemma~\ref{separationSecondTower2Robots}, we know
  that as long as $r$ is on a same node as the robots of $T"$ then $T"$ cannot 
  be broken. This implies that the three robots form a 3-long-lived tower, which
  leads to a contradiction with the fact that there is no 3-long-lived tower in
  $\mathcal{E}$. Hence this case cannot happen.
  
  \item [Case 2:]
  
  At time $t_{s} - 1$ the robots of $T$ must be edge-activated, otherwise there
  is a contradiction with the fact that $T$ starts at time $t_{s}$.
  
  Since there is no long-lived tower at time $t_{s} - 1$ then by 
  Lemma~\ref{leLemmeQuilFallait}, at time $t_{s}$ it is not possible to have a 
  tower of 3 robots. Then since at time $t_{s}$, $T$ is formed, it exists at 
  time $t_{s}$ a tower of 2 robots. For two robots to form a tower at time 
  $t_{s}$, during the Move phase of time $t_{s} - 1$, they either both move
  while considering two opposed global directions or they consider the same
  global direction but one of the robot cannot move (an edge is missing in its 
  direction). In both cases, thanks to the update of their variables  
  $Number\-Ro\-bots\-Pre\-vious\-Edge\-Acti\-va\-tion$ and 
  $Has\-Moved\-Pre\-vious\-Edge\-Acti\-va\-tion$ during the Compute phase of 
  time $t_{s} - 1$, during the Move phase of the first time greater or equal to
  $t_{s}$ where these two robots are edge-activated, they consider opposed
  global directions and separate them. Therefore there is a contradiction with
  the fact that $T$ is a 2-long-lived tower starting at time $t_{s}$.
 \end{description}
                        
 Therefore there exists a tower of 2 robots $T'$ during the Look phase of time 
 $t_{s} - 1$. Now assume, by contradiction that the two robots of $T'$ are 
 involved in $T$. If $T'$ is a 2-long-lived tower then during the Move phase of
 time $t_{s} - 1$ the two robots of $T'$ are edge-activated and consider two 
 opposed global directions, otherwise there is a contradiction with the fact
 that $T$ starts at time $t_{s}$. If $T'$ is a 2-short-lived tower then during
 the Move phase of time $t_{s} - 1$ the two robots of $T'$ are edge-activated
 (otherwise $T$ cannot be a 2-long-lived tower), and they consider two opposite 
 global directions (by definition of a 2-short-lived tower). 
 
 A robot can cross only one edge at each instant time. Since $n \geq 4$  
 whatever the situation (only one of the robots of $T$ moves or both of the
 robots of $T$ move during the Move phase of time $t_{s} - 1$) the two robots of
 $T$ cannot be again on a same node at time $t_{s}$. In conclusion, only one 
 robot of $T'$ is involved in $T$.
 
 
 Finally, assume by contradiction, that during the Move phase of time 
 $t_{s} - 1$, either both the robots of $T$ move (in this case, during the Move
 phase of time $t_{s} - 1$ the two robots consider two opposed global 
 directions otherwise they cannot meet to form $T$) or only the robot of $T$
 involved in $T'$ moves while the other robot of $T$ does not move (in this 
 case, during the Move phase of time $t_{s} - 1$ the two robots consider the 
 same global direction otherwise they cannot meet to form $T$). In both cases,
 thanks to the update of the variables 
 $Has\-Moved\-Pre\-vious\-Edge\-Acti\-va\-tion$ and 
 $Number\-Ro\-bots\-Pre\-vious\-Edge\-Acti\-va\-tion$ during the Compute phase
 of time $t_{s} - 1$, during the Move phase of the first time after time $t_{s}$
 where the robots of $T$ are edge-activated, they consider two opposed global 
 directions. Therefore there is a contradiction with the fact that $T$ is a 
 2-long-lived tower starting at time $t_{s}$.
\end{proof}

The next two lemmas show that the whole ring is visited between two consecutive 
2-long-lived towers if these two towers satisfy some properties.

\begin{lemma} \label{needToVisitTheWholeRing}
 Consider an execution $\mathcal{E}$ without any 3-long-lived tower but 
 containing a 2-long-lived tower $T = (S, [t_{s}, t_{e}])$. If there exists
 another 2-long-lived tower $T' = (S', [t_{s}', t_{e}'])$, with 
 $t_{s}' > t_{e} + 1$ and such that $T'$ is the first 2-long-lived tower after
 $T$ in $\mathcal{E}$, then all the nodes of $\mathcal{G}$ have been visited by
 at least one robot between time $t_{e}$ and time $t_{s}' - 1$.
\end{lemma}

\begin{proof}
 Consider an execution $\mathcal{E}$ without any 3-long-lived tower but 
 containing a 2-long-lived tower $T = (S, [t_{s}, t_{e}])$. Assume that there
 exists another 2-long-lived tower $T' = (S', [t_{s}', t_{e}'])$, with 
 $t_{s}' > t_{e} + 1$ and such that $T'$ is the first 2-long-lived tower after
 $T$ in $\mathcal{E}$.
 
 Since by assumption there is no long-lived tower between the Look phase of time
 $t_{e} + 1$ and the Look phase of time $t_{s}' - 1$ included, then by 
 Lemma~\ref{noThreeShortLivedTower}, from the Look phase of time $t_{e} + 1$ to 
 the Look phase of time $t_{s}' - 1$ included, if some robots meet they only 
 form 2-short-lived towers. Therefore, by Lemma~\ref{existsTower2Robots}, at 
 time $t_{s}' - 1$ there exists a 2-short-lived tower $T_{short}$.
 
 To form $T'$, by Lemma~\ref{existsTower2Robots}, the configuration $C$ reached 
 is such that $T_{short}$ and the robot of $T'$ not involved in $T_{short}$ are
 on two adjacent nodes, the adjacent edge to the location of $T_{short}$ in the 
 global direction $d$ is missing at time $t_{s}' - 1$, and the two robots of 
 $T'$ are edge-activated and consider the global direction $d$ during the Move 
 phase of time $t_{s}' - 1$. During the Move phase of time $t_{e}$ the 
 configuration $C'$ is such that the two robots of $T$ are on a same node
 considering two opposed global direction. Moreover, from the Look phase of
 time $t_{e} + 1$ to the Look phase of time $t_{s}' - 1$ included, if two robots
 meet they separate once they are edge-activated considering two opposed global
 directions. Besides, while executing \PEF, a robot does not change the global 
 direction it considers if it is isolated. All this implies that to reach $C$ 
 from $C'$ all the nodes of $\mathcal{G}$ have been visited by at least one 
 robot between time $t_{e}$ and time $t_{s}' - 1$.
\end{proof}

\begin{lemma} \label{needToVisitTheWholeRing_bis}
 Consider an execution $\mathcal{E}$ without any 3-long-lived tower, and let 
 $T_{i} = (S_{i}, [t_{s\_i}, t_{e\_i}])$ be the $i^{th}$ 2-long-lived tower of
 $\mathcal{E}$, with $i \geq 2$. If
 $T_{i + 1} = (S_{i+1},$ $[t_{s\_i+1}, t_{e\_i+1}])$ exists and satisfies
 $t_{s\_i+1} = t_{e\_i}  + 1$, then all the nodes of $\mathcal{G}$ have been 
 visited by at least one robot between time $t_{s\_i} - 1$ and time
 $t_{s\_i+1} - 1$.
\end{lemma}

\begin{proof}  
 Consider an execution $\mathcal{E}$ without any 3-long-lived tower but 
 containing a 2-long-lived tower $T_{i} = (S_{i}, [t_{s\_i}, t_{e\_i}])$, with
 $i \geq 2$. Assume that there exists another 2-long-lived tower 
 $T_{i + 1} = (S_{i+1}, [t_{s\_i+1}, t_{e\_i+1}])$, with 
 $t_{s\_i+1} = t_{e\_i}  + 1$. By Lemma~\ref{existsTower2Robots}, to form 
 $T_{i + 1}$, a tower of 2 robots involving only one robot of $T_{i + 1}$ must
 be present at time $t_{s\_i+1} - 1$. Moreover $T_{i}$ is a tower of 2 robots 
 which is present in $\mathcal{G}$ from time $t_{s\_i}$ to time 
 $t_{s\_i+1} - 1$. Therefore $S_{i+1} \neq S_{i}$.
 
 To form $T_{i}$, by Lemma~\ref{existsTower2Robots}, the configuration $C$ 
 reached at time $t_{s\_i} - 1$ is such that there is a tower $T$ of 2 robots
 involving only one robot of $T_{i}$ and the other robot of $T_{i}$ which are on
 two adjacent nodes
 
 
 Similarly, by Lemma~\ref{existsTower2Robots}, and since 
 $t_{s\_i+1} = t_{e\_i}  + 1$, to form $T_{i + 1}$, the configuration $C'$ 
 reached at time $t_{s\_i+1} - 1$ is such that $T_{i}$ and the robot of 
 $T_{i + 1}$ not involved in $T_{i}$ are on two adjacent nodes, the adjacent 
 edge to the location of $T_{i}$ in the global direction $d$ is missing at time 
 $t_{s\_i+1} - 1$, and the two robots of $T_{i + 1}$ are edge-activated and
 consider the global direction $d$ during the Move phase of time 
 $t_{s\_i+1} - 1$. Moreover, since there is no 3-long-lived tower in
 $\mathcal{E}$, from the Look phase of time $t_{s\_i}$ to the Look phase of time
 $t_{s\_i+1} - 1$ included, if $T_{i}$ meets the other robot of the system, they
 form a 3-short-lived tower and hence they separate once they are edge-activated 
 considering two opposed global directions. Besides, while executing \PEF, a 
 robot does not change the global direction it considers if it is isolated. All
 this implies that to reach $C'$ from $C$ all the nodes of $\mathcal{G}$ have 
 been visited by at least one robot between time $t_{s\_i} - 1$ and time 
 $t_{s\_i+1} - 1$.
\end{proof}

\subsection{Correctness Proof}

Upon establishing all the above properties of towers, we
are now ready to state the main lemmas of our proof. Each of these three lemmas 
below shows that after time $t_{max}$ our algorithm performs the perpetual
exploration in a self-stabilizing way for a specific subclass of 
connected-over-time rings. 

\begin{lemma} \label{static}
 \PEF~is a perpetual exploration algorithm for the class of static rings of
 arbitrary size using three robots.
\end{lemma}

\begin{proof}
 Assume that $\mathcal{G}$ is a static ring. While executing \PEF, a robot 
 considers a direction at each round. Moreover, a robot does not change the
 global direction it considers if its variable 
 $Has\-Moved\-Pre\-vious\-Edge\-Acti\-va\-tion$ is true. The variables of a 
 robot are updated during Compute phases of times where it is edge-activated.
 Since $\mathcal{G}$ is static, this implies that in each round all the robots
 are edge-activated and are able to move whatever the direction they consider.
 So, after $t_{max}$ their variables
 $Has\-Moved\-Pre\-vious\-Edge\-Acti\-va\-tion$ are always true. Hence, the
 robots never change their directions. 
 
 As $(i)$ the robots have a stable direction, $(ii)$ they always consider 
 respectively the same global direction, and $(iii)$ there always exists an 
 adjacent edge to their current locations in the global direction they consider,
 the robots move infinitely often in the same global direction. Moreover, as
 $\mathcal{G}$ has a finite size, this implies that all the robots visit 
 infinitely often all the nodes of $\mathcal{G}$.
\end{proof}

\begin{lemma} \label{SomeEdgesMissingSometimes}
 \PEF~is a perpetual exploration algorithm for the class of edge-recurrent but 
 non static rings of arbitrary size using three robots.
\end{lemma}

\begin{proof}
 Assume that $\mathcal{G}$ is an edge-recurrent but non static ring. Let us
 study the following cases.
 
 \begin{description}
  \item [Case 1:] \textbf{There exists at least one 3-long-lived tower in 
  $\mathbf{\mathcal{E}}$.}
  
   \begin{description}
    \item [Case 1.1:] \textbf{One of the 3-long-lived towers of 
    $\mathbf{\mathcal{E}}$ has an infinite duration.}
    
     Denote by $T = (S, [t_{s}, +\infty[)$ the 3-long-lived tower of
     $\mathcal{E}$ that has an infinite duration. Call $t \geq t_{s}$ the first
     time after time $t_{s}$ where the robots of $T$ are edge-activated. The 
     variables of a robot are updated during Compute phases of times where it is
     edge-activated. Therefore, since there are three robots in the system, from
     time $t_{act} + 1$, the condition ``$Number\-Of\-Ro\-bots\-On\-No\-de() 
     > Number\-Ro\-bots\-Pre\-vious\-Edge\-Acti\-va\-tion$'' is false for the 
     three robots of $T$. Therefore from time $t_{act} + 1$ the predicate
     $I\-Was\-Stuck\-On\-My\-No\-de\-And\-Now\-We\-Are\-More\-Ro\-bots()$ of
     each robot of $T$ is false.
     
     By Corollary~\ref{breakingTowerCorollary}, eventually, the predicates
     $We\-Are\-Stuck\-In\-The\-Same\-Dire\-ction()$ of the robots of $T$ are
     always false, otherwise $T$ is broken in finite time, which leads to a 
     contradiction. 
     
     Since eventually the predicates 
     $I\-Was\-Stuck\-On\-My\-No\-de\-And\-Now\-We\-Are\-More\-Ro\-bots()$ and 
     $We\-Are\-Stuck\-In\-The\-Same\-Dire\-ction()$ of the robots of $T$ are 
     always false, then eventually they consider always the same global direction. 
     $\mathcal{G}$ is edge-recurrent, therefore there exists infinitely often an
     adjacent edge to the location of $T$ in the global direction considered by
     the robots of $T$, then the robots are able to move infinitely often in the
     same global direction. Moreover, as $\mathcal{G}$ has a finite size, all
     the robots visit infinitely often all the nodes of $\mathcal{G}$.
  
    \item [Case 1.2:] \textbf{Any 3-long-lived tower of $\mathbf{\mathcal{E}}$
    has a finite duration.}
    
    By Lemma~\ref{noThreeLongLivedTower}, once a 3-long-lived tower is broken,
    it is impossible to have another 3-long-lived tower in $\mathcal{E}$. Then,
    $\mathcal{E}$ admits an infinite suffix that matches either case 2 or 3.
   \end{description}

  \item [Case 2:] \textbf{There exists at least one 2-long-lived tower in 
  $\mathbf{\mathcal{E}}$.}
   \begin{description}
    \item [Case 2.1:] \textbf{There exists a finite number of 2-long-lived 
    towers in $\mathbf{\mathcal{E}}$.}
    
    Let $T' = (S', [t_{s}', t_{e}'])$ be the last 2-long-lived tower of 
    $\mathcal{E}$.
    
    There is no 3-long-lived tower in $\mathcal{E}$ at time $t_{s}'$ (otherwise 
    Case 1 is considered), hence by Lemma~\ref{noThreeLongLivedTower} there
    is no 3-long-lived tower in $\mathcal{E}$. Moreover, if $T'$ has a finite 
    duration, then $\mathcal{E}$ admits an infinite suffix with no long-lived 
    tower, hence matching case 3.
    
    Otherwise, (\ie $T'$ has an
    infinite duration), as in Case 1.1, the robots of $T'$ eventually have their
    predicates $We\-Are\-Stuck\-In\-The\-Same\-Dire\-ction()$ always false, 
    otherwise, $T'$ is broken in finite time. Let $t_{false}$ be the time from
    which the robots of $T'$ have their predicates 
    $We\-Are\-Stuck\-In\-The\-Same\-Dire\-ction()$ always false. After time
    $t_{false}$, the only case when the robots of $T'$ change the global
    direction they consider, is when they meet the third robot of the system.
     
     \begin{description}
      \item [Case 2.1.1:] \textbf{The robots of $\mathbf{T'}$ meet the third
      robot finitely often.}
      
      After the time when the last tower of 3 robots is broken, the robots of 
      $T'$ have their predicates 
      $I\-Was\-Stuck\-On\-My\-No\-de\-And\-Now\-We\-Are\-More\-Ro\-bots()$
      always false. Let $t_{break}$ be the time when the last tower of 3 robots 
      if broken. From time $t = max \{t_{break}, t_{false}\} + 1$ the robots of
      $T'$ have their predicates
      $I\-Was\-Stuck\-On\-My\-No\-de\-And\-Now\-We\-Are\-More\-Ro\-bots()$ and 
      $We\-Are\-Stuck\-In\-The\-Same\-Dire\-ction()$ always false, therefore 
      they always consider the same global direction. Since $\mathcal{G}$ is 
      edge-recurrent, there is infinitely often an adjacent edge to the location
      of $T'$ in the direction considered by the robots of $T'$. This implies
      that they are able to move infinitely often in the same global direction.
      Moreover, as $\mathcal{G}$ has a finite size, this implies that all the
      robots visit infinitely often all the nodes of $\mathcal{G}$.
      
      \item [Case 2.1.2:] \textbf{The robots of $\mathbf{T'}$ meet the third 
      robot infinitely often.}
      
      Consider the execution after time $t_{false}$. The robot not involved in 
      $T'$ does not change its direction while it is isolated. Similarly, the 
      robots of $T'$ maintain their directions until they meet the third robot.
      Moreover, when the robots of $T'$ meet the third robot of the system, they
      form a 3-short-lived tower. Therefore once they are edge-activated, they 
      separate them considering opposed global directions. Then, we can deduce
      that all the nodes of $\mathcal{G}$ are visited between two consecutive 
      meetings of $T'$ and the third robot. As $T'$ and the third robot 
      infinitely often meet, all the nodes of $\mathcal{G}$ are infinitely often
      visited.
    
    \end{description}

    \item [Case 2.2:] \textbf{There exist an infinite number of 2-long-lived 
    towers in $\mathbf{\mathcal{E}}$.}
    
    By Lemmas~\ref{needToVisitTheWholeRing} and
    \ref{needToVisitTheWholeRing_bis}, we know that between two consecutive 
    2-long-lived towers (from the second one), all the nodes of $\mathcal{G}$ 
    are visited. As there is an infinite number of 2-long-lived towers, the
    nodes of $\mathcal{G}$ are infinitely often visited.
    
   \end{description}
  \item [Case 3:] \textbf{There exist no long-lived tower in
  $\mathbf{\mathcal{E}}$.}
   
  Then, we know, by Lemma~\ref{noThreeShortLivedTower}, that $\mathcal{E}$
  contains only configurations with either three isolated robots or one
  2-short-lived tower and one isolated robot. 
%
%

      We want to prove the following property. If during the Look phase of time
      $t$, a robot $r$ is located on a node $u$ considering the global direction
      $gd$, then there exists a time $t' \geq t$ such that, during the Look 
      phase of time $t'$, a robot is located on the node $v$ adjacent to $u$ in
      the global direction $gd$ and considers the global direction $gd$.
      
      Let $t" \geq t$ be the smallest time after time $t$ where the adjacent
      edge of $u$ in the global direction $gd$ is present in $\mathcal{G}$. As 
      all the edges of $\mathcal{G}$ are infinitely often present, $t"$ exists.
      
      
      \noindent $(i)$ If $r$ crosses the adjacent edge of $u$ in the global 
      direction $gd$ during the Move phase of time $t"$, then the property is 
      verified.
      
      \noindent $(ii)$ If $r$ does not cross the adjacent edge of $u$ in the 
      global direction $gd$ during the Move phase of time $t"$, this implies 
      that $r$ changes the global direction it considers during the Look phase
      of a time $t$. While executing \PEF, a robot can change the global 
      direction it considers only during Compute phases of times where it is 
      edge-activated and involved in a tower. Let $t_{act} \geq t$ be the first
      time after time $t$ such that during the Move phase of time $t_{act}$, $r$
      does not consider the global direction $gd$. Let $r'$ the robot involved 
      in a tower with $r$ at time $t_{act}$. Since there are only 2-short-lived 
      towers in the execution, the two robots $r$ and $r'$ consider two opposed 
      global directions during the Move phase of time $t_{act}$. Therefore 
      during the Move phase of time $t_{act}$, $r'$ is on node $u$ considering
      the global direction $gd$. By applying case $(ii)$ by 
      recurrence, we can say that from the Move phase of time $t$ to the Move 
      phase of time $t"$ there always exists a robot on node $u$ considering the 
      global direction $gd$. Therefore during the Move phase of time $t"$ a 
      robot moves on node $v$. Since the robot does not change the global 
      direction they consider during Look phases, during the Look phase of time
      $t" + 1$ this robot still considers the global direction $gd$.      
    
      This prove the property. By applying recurrently this property to any 
      robot, we prove that all the nodes are infinitely often visited.
      
  
 \end{description}
Thus, we obtain the desired result in every cases.
\end{proof}

\begin{lemma} \label{OneEdgeMissing}
 \PEF~is a perpetual exploration algorithm for the class of connected-over-time
 but not edge-recurrent rings of arbitrary size using three robots.
\end{lemma}

\begin{proof}
 Consider that $\mathcal{G}$ is a connected-over-time but not edge-recurrent 
 ring. This implies that there exists exactly one eventual missing edge $e$ in 
 $\mathcal{G}$. Denote by $\mathcal{E}^1$ the maximal suffix of $\mathcal{E}$ in
 which the eventual missing edge never appears. Let $t_{missing}$ the time after 
 which $e$ never appears again. Let us study the following cases.
 
 \begin{description}
  \item [Case 1:] \textbf{There exists at least one 3-long-lived tower in 
  $\mathbf{\mathcal{E}^{1}}$.}
  
  According to Lemma~\ref{OneEdgeMissingTowerBroken}, this 3-long-liver tower is 
  broken in finite time. Moreover, once this tower is broken, according to
  Lemma~\ref{noThreeLongLivedTower}, it is impossible to have a configuration 
  containing a 3-long-lived tower. Then, $\mathcal{E}^{1}$ admits an infinite 
  suffix that matches either case 2 or 3.

  \item [Case 2:] \textbf{There exists at least one 2-long-lived tower in 
  $\mathbf{\mathcal{E}^{1}}$.}
   \begin{description}
    \item [Case 2.1:] \textbf{There exists a finite number of 2-long-lived 
    towers in $\mathbf{\mathcal{E}^{1}}$.}
    
    According to Lemma~\ref{OneEdgeMissingTowerBroken}, the last 2-long-lived
    tower is broken in finite time. Since by Lemma~\ref{noThreeLongLivedTower},
    it cannot exists 3-long-lived tower in $\mathcal{E}^{1}$, then
    $\mathcal{E}^{1}$ admits an infinite suffix with no long-lived tower hence
    matching Case 3.

    \item [Case 2.2:] \textbf{There exist an infinite number of 2-long-lived 
    towers in $\mathbf{\mathcal{E}^{1}}$.}
    
    By Lemmas~\ref{needToVisitTheWholeRing} and
    \ref{needToVisitTheWholeRing_bis}, we know that between two consecutive 
    2-long-lived towers (from the second one), all the nodes of $\mathcal{G}$ 
    are visited. As there is an infinite number of 2-long-lived towers, all the
    nodes of $\mathcal{G}$ are infinitely often visited.
   \end{description}  
   
  \item [Case 3:] \textbf{There exist no long-lived tower in
  $\mathbf{\mathcal{E}^{1}}$.}
   
  By Lemma~\ref{noThreeShortLivedTower}, all configurations in $\mathcal{E}^{1}$
  contain either three isolated robots or one 2-short-lived tower and one 
  isolated robot. 

  \noindent $(1)$
  We want to prove the following property. If during the Look phase of a 
  time $t$ in $\mathcal{E}^{1}$, a robot considers a global direction $gd$
  and is located on a node at a distance $d \neq 0$ in 
  $G$ ($G$ is the footprint of $\mathcal{G}$) from the extremity of $e$ in the 
  global direction $gd$, then it exists a time $t' \geq t$ such that, during the
  Look phase of time $t'$, a robot is on a node at distance $d - 1$ in
  $G$ from the extremity of $e$ in the global direction 
  $gd$ and considers the global direction $gd$. Let $v$ be the adjacent node of 
  $u$ in the global direction $gd$.
      
  Let $t" \geq t$ be the smallest time after time $t$ where the adjacent edge of
  $u$ in the global direction $gd$ is present in $\mathcal{G}$. As all the edges
  of $\mathcal{G}$ except $e$ are infinitely often present and as $u$ is at a
  distance $d \neq 0$ in $G$ from the extremity of $e$ in the
  global direction $gd$, then the adjacent edge of $u$ in the global direction
  $gd$ is infinitely often present in $\mathcal{G}$. Hence, $t"$ exists.
 
      
  \noindent $(i)$ If $r$ crosses the adjacent edge of $u$ in the global 
  direction $gd$ during the Move phase of time $t"$, then the property is 
  verified.
  
  \noindent $(ii)$ If $r$ does not cross the adjacent edge of $u$ in the 
  global direction $gd$ during the Move phase of time $t"$, this implies 
  that $r$ changes the global direction it considers during the Look phase
  of a time $t$. While executing \PEF, a robot can change the global 
  direction it considers only during Compute phases of times where it is 
  edge-activated and involved in a tower. Let $t_{act} \geq t$ be the first
  time after time $t$ such that during the Move phase of time $t_{act}$, $r$
  does not consider the global direction $gd$. Let $r'$ the robot involved 
  in a tower with $r$ at time $t_{act}$. Since there are only 2-short-lived 
  towers in the execution, the two robots $r$ and $r'$ consider two opposed 
  global directions during the Move phase of time $t_{act}$. Therefore 
  during the Move phase of time $t_{act}$, $r'$ is on node $u$ considering
  the global direction $gd$. By applying case $(ii)$ by 
  recurrence, we can say that from the Move phase of time $t$ to the Move 
  phase of time $t"$ there always exists a robot on node $u$ considering the 
  global direction $gd$. Therefore during the Move phase of time $t"$ a 
  robot moves on node $v$. Since the robot does not change the global 
  direction they consider during Look phases, during the Look phase of time
  $t" + 1$ this robot still considers the global direction $gd$.      

  This prove the property. 
  
  \noindent $(2)$
  We now want to prove that there exists a time $t_{reachExtremities}$ in 
  $\mathcal{E}^{1}$ from which one robot is forever located on each extremity of
  $e$ pointing to $e$.
      
 First, we want to prove that a robot reaches one of the extremities of $e$ 
 in a finite time after $t_{missing}$ and points to $e$ at this time. If it is 
 not the case at time $t_{missing}$, then there exists at this time a robot
 considering a global direction $gd$ and located on a node $u$ at distance
 $d \neq 0$ in $G$ from the extremity of $e$ in the global
 direction $gd$. By applying $d$ times the property $(1)$, we prove that, during
 the Look phase of a time $t_{reach} \geq t_{missing}$, a robot (denote it $r$) 
 reaches the extremity of $e$ in the global direction $gd$ from $u$ (denote it
 $v$ and let $v'$ be the other extremity of $e$), and that this robot considers
 the global direction $gd$ during the Look phase of time $t_{reach}$. 
 
 Then, we can prove that from time $t_{reach}$ there always exists a robot on
 node $v$ considering the global direction $gd$. Indeed, note that no robot can
 cross $e$ in the global direction $gd$ from time $t_{reach}$ since $e$ is 
 missing from time $t_{missing}$. Moreover while executing \PEF, a robot can 
 change the global direction it considers only during Compute phases of times
 where it is edge-activated and involved in a tower. Therefore if at a time
 $t_{change} \geq t_{missing}$, $r$ changes the global direction it considers at
 time $t_{reach}$ this is because it is involved in a tower. Since there are 
 only 2-short-lived towers in the execution, at time $t_{change}$, $r$ is
 involved in a tower with a robot $r'$, and $r$ and $r'$ consider two opposed
 global directions during the Move phase of time $t_{change}$. Therefore during 
 the Move phase of time $t_{change}$, $r'$ is on node $v$ considering the global
 direction $gd$. By applying this argument by recurrence, we can say that from
 time $t_{reach}$ there always exists a robot on node $v$ considering the global
 direction $gd$. 
  
 Now we prove that this is also true for the extremity $v'$ of $e$. If there 
 exists at time $t_{reach}$ a robot on node $v'$ considering the global 
 direction $\overline{gd}$, or if it exists a robot considering the global 
 direction $\overline{gd}$ on a node $u'$ at distance $d \neq 0$ in 
 $G$ from $v'$ in the global direction $\overline{gd}$, then by using 
 similar arguments than the one used for $v$, we can prove the property $(2)$. If
 this is not the case, this implies that at time $t_{reach}$ all the robots
 consider the global direction $gd$. Then in finite time (after time
 $t_{reach}$) by the property $(1)$, a robot reaches node $v$. Since from time 
 $t_{reach}$ there is always a robot on node $v$, there is a 2-short-lived tower
 formed. Then by definition of a 2-short-lived tower, there exists a time at
 which one of the robots of this tower considers the global direction $gd$ while 
 the other considers the global direction $\overline{gd}$. Then we can use the 
 same arguments as the one used previously to prove the property $(2)$.
      
 \noindent $(3)$ It stays to prove that in the Case 3 all the nodes are 
 infinitely often visited. We know that from time $t_{reachExtremities}$ one 
 robot is forever located on each extremity of $e$ pointing to $e$. Call $r"$ 
 the robot that is not on node $v$ (resp. $v'$) and pointing to $e$ at time 
 $t_{reachExtremities}$. Assume that at time $t_{reachExtremities}$, $r"$ is on
 node $u'$ and considers the global direction $gd$. Then by applying recurrently
 the property $(1)$ we can prove that, in finite time, all the nodes between the
 current node of $r"$ at time $t_{reachExtremities}$ and $v$ in the global
 direction $gd$ are visited and that $r"$ reaches $v$. Call 
 $t_{act}' \geq t_{reachExtremities}$, the first time after time
 $t_{reachExtremities}$ where there are two robots on node $v$ that are 
 edge-activated. 
 At time $t_{act}'$, the robot that is on node $v$ and pointing to $e$ at time 
 $t_{reachExtremities}$ changes the global direction it considers (hence
 considers $\overline{gd}$) by construction of \PEF~and since the tower formed 
 is a 2-short-lived tower. 

 We can then repeat this reasoning (with $v$ and $v'$ alternatively in the role 
 of $u'$ and with $v'$ and $v$ alternatively in the role of $v$) and prove that
 all nodes are infinitely often visited.      
      

 \end{description}
Thus, we obtain the desired result in every cases.
\end{proof}

To conclude the proof, first note that even if the robots can start in a non 
coherent state, it exists a time $t_{max}$ from which all the robots of the 
system are in a coherent state (by Lemma~\ref{convergence}). Then it is 
sufficient to observe that a connected-over-time ring is by definition either 
static, edge-recurrent but non static, or connected-over-time but not
edge-recurrent. As we prove the correctness of our algorithm from the 
time the robots are in a coherent state in these three cases in 
Lemmas~\ref{static}, \ref{SomeEdgesMissingSometimes}, and \ref{OneEdgeMissing}
respectively, we can claim the following final result.

\begin{theorem}
 \PEF~is a self-stabilizing perpetual exploration algorithm for the class of
 connected-over-time rings of arbitrary size (greater or equal to four) using
 three robots with distinct identifiers.
\end{theorem}

\section{Sufficiency of Two Robots for $n=3$}\label{sec:algo2robots}

   In this section, we present \PEFR, a self-stabilizing algorithm solving 
deterministically the perpetual exploration problem on connected-over-time rings
of size equal to 3, using two robots possessing distinct identifiers.

To present this algorithm we add a new predicate, named 
$I\-Am\-Stuck\-Alone\-On\-My\-Node()$ defined as follows.

\vspace{0.2cm}
\noindent\begin{tabular}{l@{}l@{}l}
$I\-Am\-Stuck\-Alone\-On\-My\-Node()\equiv$\\
& & $(NumberOfRobotsOnNode() = 1)$\\
& $\wedge$ & $\lnot ExistsEdgeOnCurrentDirection()$ \\
& $\wedge$ & $ExistsEdgeOnOppositeDirection()$
\end{tabular}
\vspace{0.1cm}

\begin{algorithm}
\caption{\PEFR} \label{SSPER}
\footnotesize
    \begin{algorithmic} [1]
            \If {$WeAreStuckInTheSameDirection()$} 
                \State \Call{GiveDirection}{}
            \EndIf
            \If {$I\-Am\-Stuck\-Alone\-On\-My\-Node()$}
		\State \Call{OppositeDirection}{}
		
	    \EndIf
	    \State \Call{Update}{}
    \end{algorithmic}
\end{algorithm}

The pseudo-code of \PEFR~ is given in Algorithm~\ref{SSPER}. 


\paragraph{Proof of correctness.} We now prove the correctness of this 
algorithm.

First, note that Lemmas~\ref{convergence}, \ref{separation_possible},
\ref{separation_possible_bis} are also true for \PEFR.

To show the correctness of \PEFR, we need to introduce some lemmas. We consider
that the two robots executing \PEFR~are $r_{1}$ and $r_{2}$. Let $t_{1}$, and 
$t_{2}$ be respectively the time at which the robot $r_{1}$ and $r_{2}$ are 
in a coherent state. Let $t_{max} = max\{t_{1}, t_{2}\}$. From 
Lemma~\ref{convergence}, the two robots are in a coherent state from $t_{max}$. 
In the remaining of the proof, we focus on the suffix of the execution after
$t_{max}$. The other notations correspond to the ones introduced in
Section~\ref{sec:algo3robots}.

\begin{lemma} \label{noTwoLongLivedTowerAnymore}
 Every execution starting from a configuration without a 2-long-lived tower
 cannot reach a configuration with a 2-long-lived tower.
\end{lemma}

\begin{proof}
 Assume that $\mathcal{E}$ starts from a configuration which does not contain a
 2-long-lived tower. By contradiction, let $C$ be the first configuration of 
 $\mathcal{E}$ containing a 2-long-lived tower $T = (S, [t_{s}, t_{e}])$.

 Let $t_{act} \geq t_{s}$ be the first time after time $t_{s}$ where the 2 
 robots of $T$ are edge-activated. By definition of a long-lived tower,
 $t_{act}$ exists.
 
 For a 2-long-lived tower to be formed at time $t_{s}$, $r_{1}$ and $r_{2}$ must 
 meet at time $t_{s}$. While executing \PEFR, the two robots can meet at time 
 $t_{s}$ only because they are moving considering opposed global directions 
 during the Move phase of time $t_{s} - 1$. Therefore, since the variables of a
 robot are updated only during Compute phases of time where it is
 edge-activated, during the Look phase of time $t_{act}$, the predicates 
 $We\-Are\-Stuck\-In\-The\-Same\-Direction()$ of the two robots are false (since their
 variables $Has\-Moved\-Pre\-vious\-Edge\-Acti\-va\-tion$ are true). Moreover, 
 during the Look phase of time $t_{act}$ the predicates
 $I\-Am\-Stuck\-Alone\-On\-My\-Node()$ of the two robots are false (since their 
 predicates $NumberOfRobotsOnNode()$ is not equal to 1). Hence during the Move 
 phase of time $t_{act}$ the two robots still consider two opposed global 
 directions. Therefore $T$ is broken at time $t_{act}$, which leads to a 
 contradiction with the fact that $T$ is a 2-long-lived tower. This proves the 
 lemma.
\end{proof}

Let $t_{act1}$ (resp. $t_{act2}$) be the first time in the execution at which 
the robot $r_{1}$ (resp. $r_{2}$) is edge-activated. By definition, we have 
$t_{1} = t_{act1} + 1$ and $t_{2} = t_{act2} + 1$. By
Lemma~\ref{noTwoLongLivedTowerAnymore}, if there exists a 2-long-lived tower in
$\mathcal{E}$, then this 2-long-lived tower is present in the execution from 
time $t_{0} = 0$. In this case $t_{1} = t_{2} = t_{max}$ and at time
$t_{max} - 1$ the robots are edge-activated for the first time of the execution.

\begin{lemma} \label{propertyLongLivedTowerBis}
 The robots of a long-lived tower $T = (S, [t_{s}, t_{e}])$ consider a same 
 global direction at each time between the Look phase of round $t_{max}$ and the
 Look phase of round $t_{e}$ included.
\end{lemma}

\begin{proof}
 Consider a long-lived tower $T = (S, [t_{s}, t_{e}])$. We know that 
 $t_{s} = t_{0} = 0$, that $t_{1} = t_{2} = t_{max}$ and that at time 
 $t_{max} - 1$ the robots are edge-activated for the first time of the 
 execution. During the Move phase of time $t_{max} - 1$, the two robots consider
 the same global direction, otherwise there is a contradiction with the fact
 that $T$ is a 2-long-lived tower.

 When executing \PEFR, a robot can change the global direction it considers only
 when it is edge-activated. Besides, during the Look phase of a time $t$ a robot
 considers the same global direction than the one it considers during the Move
 phase of time $t - 1$.
 
 Consider a time $t \in [t_{max}, t_{e}[$. If at time $t$ the robots of $S$ are 
 not edge-activated, then during the Move phase of time $t$ the robots of $S$ do
 not change the global direction they consider. 
 
 If at time $t$ the robots of $S$ are edge-activated, then during the Move phase
 of time $t$, since $t \neq t_{e}$, the robots of $S$ consider the same global
 direction, otherwise there is a contradiction with the fact that $T$ is a 
 long-lived tower from time $t_{s}$ to time $t_{e}$.
 
 Since at time $t_{max} - 1$ the robots of $S$ consider the same global
 direction using the two previous arguments by recurrence on each time 
 $t \in [t_{max}, t_{e}[$ and the fact that robots change the global directions
 they consider only during Compute phases, we can conclude that the robots of 
 $S$ consider a same global direction from the Look phase of time $t_{max}$ to 
 the Look phase of time $t_{e}$ included.
\end{proof}

\begin{lemma} \label{successive_bits_consideration_bis}
 For any long-lived tower $T = (S, [t_{s}, t_{e}])$, and any $t \leq t_{e}$,
 such that the robots of $S$ have been edge-activated twice between $t_{s}$ 
 included and $t$ not included, we have
 $We\-Are\-Stuck\-In\-The\-Same\-Dire\-ction()(r_{1}, t)$ $=$ 
 $We\-Are\-Stuck\-In\-The\-Same\-Dire\-ction()(r_{2}, t)$.
\end{lemma}

\begin{proof}
 Consider a long-lived tower $T = (S, [t_{s}, t_{e}])$. We know that 
 $t_{s} = t_{0} = 0$, that $t_{1} = t_{2} = t_{max}$ and that at time 
 $t_{max} - 1$ the robots are edge-activated for the first time of the 
 execution. Assume that between $t_{s}$ included and $t_{e}$ not included, the 
 robots of $T$ are edge-activated two or more times.
 
 By definition of a long-lived tower and by
 lemma~\ref{propertyLongLivedTowerBis}, from the Look phase of time $t_{max}$
 to the Look phase of time $t_{e}$ included, all the robots of $S$ are on a same 
 node and consider a same global direction. Therefore the values of their 
 respective predicates $Number\-Of\-Ro\-bots\-On\-No\-de()$, 
 $Exists\-Edge\-On\-Current\-Dire\-ction()$ and
 $Exists\-Edge\-On\-Oppo\-si\-te\-Dire\-ction()$ are identical from the Look 
 phase of time $t_{max}$ to the Look phase of time $t_{e}$ included.
 
 Let $t_{act} \geq t_{max}$ be the first time after $t_{max}$ such that the 
 robots of $T$ are edge-activated. By assumption, $t_{act}$ exists. When 
 executing \PEFR, a robot updates its variables
 $Has\-Moved\-Pre\-vious\-Edge\-Acti\-va\-tion$ and 
 $Number\-Ro\-bots\-Pre\-vious\-Edge\-Acti\-va\-tion$ respectively with the 
 values of its predicates $Exists\-Edge\-On\-Current\-Dire\-ction()$ and 
 $Number\-Of\-Ro\-bots\-On\-No\-de()$, only during Compute phases of times
 when it is edge-activated. Therefore, from the Look phase of time $t_{act} + 1$
 to the Look phase of time $t_{e}$ included, the robots of $S$ have the same
 values for their variables $Has\-Moved\-Pre\-vious\-Edge\-Acti\-va\-tion$ and 
 $Number\-Ro\-bots\-Pre\-vious\-Edge\-Acti\-va\-tion$.

 The predicate $We\-Are\-Stuck\-In\-The\-Same\-Dire\-ction()$ depends only on
 the values of the variables $Has\-Moved\-Pre\-vious\-Edge\-Acti\-va\-tion$, 
 $Number\-Ro\-bots\-Pre\-vious\-Edge\-Acti\-va\-tion$ and on the values of the 
 predicates $Number\-Of\-Ro\-bots\-On\-No\-de()$, 
 $Exists\-Edge\-On\-Current\-Dire\-ction()$, and 
 $Exists\-Edge\-On\-Oppo\-si\-te\-Dire\-ction()$. As seen previously, all these
 values are identical for all the robots of $S$ from the Look phase of time
 $t_{act} + 1$ until the Look phase of time $t_{e}$ included. This prove the
 lemma.
\end{proof}

From the Lemmas~\ref{successive_bits_consideration_bis}, 
\ref{separation_possible} and \ref{separation_possible_bis}, by noticing that
the robots of a long-lived tower $T$ cannot have their predicates
$I\-Am\-Stuck\-Alone\-On\-My\-Node()$ true as long as their are involved in $T$,
we can again obtain the corollary~\ref{breakingTowerCorollary} (the proof
is not exactly the same since the predicate 
$I\-Was\-Stuck\-On\-My\-No\-de\-And\-Now\-We\-Are\-More\-Ro\-bots()$ does not 
exist in \PEFR, however the proof is very similar, therefore not repeated in 
this section).

\begin{theorem}\label{th:algo2robots}
 \PEFR~is a deterministic self-stabilizing perpetual exploration algorithm for
 the class of connected-over-time rings of size equals to $3$ using 2 
 fully synchronous robots possessing distinct identifiers.
\end{theorem}

\begin{proof}
 Consider that $\mathcal{G}$ is a connected-over-time ring of size $3$. First
 note that even if the robots can start in a non coherent state, 
 by Lemma~\ref{convergence}, it exists a time $t_{max}$ from which all the
 robots are in a coherent state. Let us study the following cases occurring when
 the robots are in a coherent state.
 
 \begin{description}
  \item [Case 1 :] \textbf{There exists at least one 2-long-lived tower in 
  $\mathbf{\mathcal{E}}$.}
  
  By Lemma~\ref{noTwoLongLivedTowerAnymore}, once a 2-long-lived tower is
  broken, it is not possible to have again a 2-long-lived tower in
  $\mathcal{E}$. Therefore there exists only one 2-long-lived tower in 
  $\mathcal{E}$.
 
    If the 2-long-lived tower of $\mathcal{E}$ has a finite duration, then
    by Lemma~\ref{noTwoLongLivedTowerAnymore}, $\mathcal{E}$ admits an infinite 
    suffix with no long-lived tower hence matching Case 2.
    
    If the 2-long-lived tower $T$ of $\mathcal{E}$ has an infinite 
    duration, the robots of $T$ eventually have their predicates 
    $We\-Are\-Stuck\-In\-The\-Same\-Dire\-ction()$ always false, otherwise,
    by Corollary~\ref{breakingTowerCorollary} $T$ is broken in finite time,
    which leads to a contradiction. Let $t_{false}$ be the time from which the 
    robots of $T$ have their predicates 
    $We\-Are\-Stuck\-In\-The\-Same\-Dire\-ction()$ always
    false. After time $t_{false}$ the robots of $T$ always consider the same 
    global direction (since their predicates 
    $I\-Am\-Stuck\-Alone\-On\-My\-Node()$ cannot be true). Moreover, after time 
    $t_{false}$ there exists infinitely often an adjacent edge to the location 
    of $T$ in the global direction considered by the robots of $T$, otherwise
    there exists a time after $t_{false}$ when the predicates 
    $We\-Are\-Stuck\-In\-The\-Same\-Dire\-ction()$ of the robots of $T$ are
    true, which is a contradiction. Hence after time $t_{false}$ the robots of 
    $T$ are infinitely often able to move in the same global direction. Since
    $\mathcal{G}$ has a finite size, all the robots visit infinitely often all 
    the nodes of $\mathcal{G}$.

  \item [Case 2:] \textbf{There exist no long-lived tower in
  $\mathbf{\mathcal{E}}$.}
   
   If there is no long-lived tower, this implies that if a tower is formed, then
   it is a 2-short-lived tower. By the connected-over-time assumption, each node
   has at least one adjacent edge infinitely often present. This implies that 
   any short-lived tower is broken in finite time. Two cases are now possible.
   
  \begin{description}
    \item [Case 2.1:] \textbf{There exists infinitely often a 2-short-lived tower
    in the execution.}
    

    Note that, if a tower is formed at a time
    $t$, then the three nodes have been visited between time $t - 1$ and time
    $t$. Then, the three nodes are infinitely often visited by a robot in the
    case where there exists infinitely often a 2-short-lived tower in the
    execution.

    \item [Case 2.2:] \textbf{There exists a time $\mathbf{t_{isolated}}$ after
    which the robots are always isolated.}

    By contradiction, assume that there
    exists a time $t'$ such that a node $u$ is never visited after $t'$. This
    implies that, after time $max\{t_{isolated}, t'\}$, either the robots are 
    always switching their position or they stay on their respective nodes. 

    In the first case, during the Look phase of each time greater than
    $max\{t_{isolated}, t'\}$, the respective variables $dir$ of the two robots
    contain the direction leading to $u$ (since each robot previously moves in 
    this direction). As at least one of the adjacent edges of $u$ is infinitely 
    often present, a robot crosses it in a finite time, that is contradictory
    with the fact that $u$ is not visited after $t'$.

    The second case implies that both adjacent edges to the location of both 
    robots are always absent after time $t_{isolated}$ (since an isolated robot
    moves as soon as it is possible, by definition of the predicate 
    $I\-Am\-Stuck\-Alone\-On\-My\-Node()$), that is contradictory with the
    connected-over-time assumption.
   \end{description}
 \end{description}
Thus, we obtain the desired result in every cases.
\end{proof}

\section{Conclusion}\label{sec:conclu}

   In this paper, we addressed the open question: ``What is the minimal size of a  
swarm of self-stabilizing robots to perform perpetual exploration of highly dynamic graphs?''. 
We give a first answer to this question by exhibiting the necessary
and sufficient numbers of such robots
to perpetually explore any connected-over-time ring, \ie any dynamic ring with 
very weak assumption on connectivity: every node is 
infinitely often reachable from any another one without any recurrence, periodicity, 
nor stability assumption. 
More precisely, we showed that necessary and sufficient numbers of robots proved in
\cite{BDP17} in a fault-free setting (2 robots for rings of size 3 and 3 robots for
rings of size greater than 4) still hold in the self-stabilizing setting
at the price of the loss of anonymity of robots.

In addition to the above contributions, our results overcome the robot networks
state-of-the-art in a couple of ways. First, at the exception of the algorithms from
\cite{BDP17}, it is the only algorithms dealing
with highly dynamic graphs. All previous solutions made some assumptions on 
periodicity or on all-time connectivity of the graph.
Second, it is the first self-stabilizing algorithm for the problem of 
exploration, either for static or for dynamic graphs.

This work opens an interesting field of research with numerous open questions.
First, we should investigate the necessity of every assumption made in this 
paper. For example, we assumed that robots are synchronous. Is this problem 
solvable with asynchronous robots? 
Second, it would be worthwhile to explore other problems in this rather 
complicated environment, \eg gathering, leader election, \etc. It may also
be interesting to consider other classes of dynamic graphs and other classes of
faults, \eg crashes of robots, Byzantine failures, \etc.

\bibliographystyle{plain}
\bibliography{biblio}
\end{document}

%
%
%
%
%

%
%
%
%
%
%
%
%
%
%
%
%
%
%
%